\RequirePackage{fix-cm}
\documentclass[smallextended]{svjour3}       
\smartqed  
\usepackage{graphicx}
\usepackage{algpseudocode}
\usepackage{algorithm}
\usepackage{algorithmicx}
\usepackage{verbatim}
\usepackage{mathtools}
\usepackage{multirow}
\usepackage{colortbl}
\usepackage{balance}
\usepackage{hyperref}
\usepackage[table]{xcolor}
\usepackage{subcaption}

\usepackage{dirtree}

\usepackage{natbib}
\usepackage{amssymb}

\usepackage{tikz}
\usetikzlibrary{fit,automata,positioning}
\usetikzlibrary{decorations.pathmorphing}
\tikzset{snake it/.style={decorate, decoration=snake}}
\usepackage{booktabs} 

\usepackage{kbordermatrix} 
\renewcommand{\kbldelim}{(} 
\renewcommand{\kbrdelim}{)}

\newlength\Origarrayrulewidth

\newcommand\mc{\multicolumn{1}{c}{\cellcolor{lightgray}\textbf{1}}}
\newcommand{\term}[1]{\emph{#1}}
\algtext*{EndWhile}
\algtext*{EndIf}
\algtext*{EndFor}
\algtext*{EndFunction}

\begin{document}

\title{One Algorithm to Evaluate Them All: Unified Linear Algebra Based Approach to Evaluate Both Regular and Context-Free Path Queries
}


\author{Ekaterina~Shemetova  \thanks {The research was supported by the Russian Science Foundation, grant number: 18-11-00100}       \and
        Rustam~Azimov \and
            Egor~Orachev \and
             Ilya~Epelbaum \and
             Semyon~Grigorev
}


\institute{E. Shemetova \at
              Saint Petersburg State University, 7/9 Universitetskaya nab., St. Petersburg, Russia\\
              Saint Petersburg Academic University, 8/3 Khlopin St., St. Petersburg, Russia\\
              JetBrains Research, Primorskiy prospekt 68-70, Building 1, St. Petersburg, Russia\\
              \email{katyacyfra@gmail.com}           
           \and
R. Azimov \at
              Saint Petersburg State University, 7/9 Universitetskaya nab., St. Petersburg, Russia\\
              JetBrains Research, Primorskiy prospekt 68-70, Building 1, St. Petersburg, Russia\\
              \email{rustam.azimov19021995@gmail.com}           
           \and
E. Orachev \at
              Saint Petersburg State University, 7/9 Universitetskaya nab., St. Petersburg, Russia\\
              JetBrains Research, Primorskiy prospekt 68-70, Building 1, St. Petersburg, Russia\\
              \email{egor.orachev@gmail.com}           
           \and
I. Epelbaum \at
              Saint Petersburg State University, 7/9 Universitetskaya nab., St. Petersburg, Russia\\
              JetBrains Research, Primorskiy prospekt 68-70, Building 1, St. Petersburg, Russia\\
              \email{iliyepelbaun@gmail.com}           
           \and
S. Grigorev \at
              Saint Petersburg State University, 7/9 Universitetskaya nab., St. Petersburg, Russia\\
              JetBrains Research, Primorskiy prospekt 68-70, Building 1, St. Petersburg, Russia\\
              \email{s.v.grigoriev@spbu.ru, semyon.grigorev@jetbrains.com}           
           \and
}

\date{Received: date / Accepted: date}

\maketitle

\begin{abstract}
  The Kronecker product-based algorithm for context-free path querying (CFPQ) was proposed by \citet{10.1007/978-3-030-54832-2_6}. We reduce this algorithm to operations over Boolean matrices and extend it with the mechanism to extract all paths of interest. We also prove $O(n^3/\log{n})$ time complexity of the proposed algorithm, where $n$ is a number of vertices of the input graph. Thus, we provide the alternative way to construct a slightly subcubic algorithm for CFPQ which is based on linear algebra and incremental transitive closure
  (a classic graph-theoretic problem), as opposed to the algorithm with the same complexity proposed by~\cite{10.1145/1328438.1328460}. Our evaluation shows that our algorithm is a good candidate to be the universal algorithm for both regular and context-free path querying.
\end{abstract}

\begin{acknowledgements}
We are grateful to Ekaterina Verbitskaia for careful reading and pointing out some mistakes.
\end{acknowledgements}
\section{Introduction}

Language-constrained path querying~\citep{barrett2000formal} is a technique for graph navigation querying.
This technique allows one to use formal languages as constraints on paths in edge-labeled graphs: a path satisfies constraints if the labels along it form a word from the specified language.

The utilization of regular languages as constraints, or \textit{Regular Path Querying} (RPQ), is the most well-studied and widespread.
Different aspects of RPQs are actively studied in graph databases~\citep{10.1145/2463664.2465216, 10.1145/3104031,10.1145/2850413}, while regular constraints are supported in such popular query languages as PGQL~\citep{10.1145/2960414.2960421} and SPARQL\footnote{Specification of regular constraints in SPARQL property paths: \url{https://www.w3.org/TR/sparql11-property-paths/}. Access date: 07.07.2020.}~\citep{10.1007/978-3-319-25007-6_1} (known as property paths).
Nevertheless, there is certainly room for improvement of RPQ efficiency, and new solutions are being created~\citep{Wang2019,10.1145/2949689.2949711}.

At the same time, using more powerful languages  as constraints, namely context-free languages, has gained popularity in recent years.
\textit{Context-Free Path Querying} problem (CFPQ) was introduced by~\cite{Yannakakis}, and nowadays is used in many areas.
For example, CFPQ is used for interprocedural static code analysis~\citep{10.1145/3158118,10.5555/271338.271343, YanSCA, Zheng:2008:DAA:1328897.1328464}
In this area CFPQ is known as the context-free language reachability (\textit{the CFL-reachability}) problem.
Also, CFPQ can be used for biological data analysis~\citep{GraphQueryWithEarley}, graph segmentation in data provenance analysis~\citep{8731467}, and for data flow information preserving in machine learning based solutions for code analysis problems~\citep{10.1145/3428301}.

Many algorithms for CPFQ were proposed, but recently~\cite{Kuijpers:2019:ESC:3335783.3335791} showed that the state-of-the-art CFPQ algorithms are still not performant enough for practical use.
This motivates further research of the new algorithms for CFPQ.

One promising way to achieve high-performance solutions for graph analysis problems is to reduce them to linear algebra operations.
To facilitate this approach, the description of basic linear algebra primitives GraphBLAS~API~\citep{7761646} was proposed.
Evaluation of the libraries that implement this API, such as SuiteSparce~\citep{10.1145/3322125} and CombBLAS~\citep{10.1177/1094342011403516}, show that reduction to linear algebra is a good way to utilize high-performance parallel and distributed computations for graph analysis.

\cite{Azimov:2018:CPQ:3210259.3210264} showed how to reduce CFPQ to matrix multiplication.
Later, it was shown by~\cite{Mishin:2019:ECP:3327964.3328503} and~\cite{10.1145/3398682.3399163} that by using the appropriate libraries for linear algebra for Azimov's algorithm implementation one can create a practical solution for CFPQ.
However, Azimov's algorithm requires transforming of the input grammar to Chomsky Normal Form.
This leads to the grammar size increase and hence worsens performance, especially for regular queries and complex context-free queries.

To solve these problems, an algorithm based on automata intersection was proposed~\citep{10.1007/978-3-030-54832-2_6}.
This algorithm is based on linear algebra and does not require the transformation of the input grammar.
In this work, we improve this algorithm by reducing it to operations over Boolean matrices, thus simplifying its description and implementation.
Additionally, we added the support of all-paths query semantics.
Under the \textit{all-path query semantics}, a query is evaluated to all paths satisfying the conditions of the query.
All-paths semantics is necessary, for example, in biological data analysis~\citep{GraphQueryWithEarley}, where paths prove why the specified vertices are of interest (similar).
In static code analysis (e.g. alias analysis) paths indicate the reason why two names are aliases which can be used to generate a good error message to the user of the static analysis tool. Reporting all such reasons (all paths) makes for a shorter feedback loop as well as provides a more detailed analysis.

We also show that this algorithm is performant enough for regular queries, so it is a good candidate for integration with the real-world query languages: one algorithm can be used to evaluate both regular and context-free queries.
Having a unified environment simplifies the development of the querying tools by allowing for reuse of common optimizations for the querying algorithm.
Note that a real-world context-free query is likely to have a regular subquery which can be significant in size.
Our algorithm is capable to treat such regular subparts as a regular query thus imposing little overhead as compared to treating them as a generic context-free query.
This makes a unified solution more promising in terms of performance.

Moreover, we show that this algorithm opens the way to tackle a long-standing problem about the existence of truly-subcubic $O(n^{3-\epsilon})$ CFPQ algorithm ~\citep{10.1145/1328438.1328460, Yannakakis}.
Currently, the best result is an $O(n^3/\log{n})$ algorithm of~\cite{10.1145/1328438.1328460}.
Also, there exist truly subcubic solutions that use fast matrix multiplication for some fixed subclasses of context-free languages~\citep{8249039}.
Unfortunately, these solutions cannot be generalized to arbitrary CFPQs.
In this work, we identify incremental transitive closure as a bottleneck on the way to achieve subcubic time complexity for CFPQ.

To sum up, we make the following contributions.
\begin{enumerate}
	\item We rethink and improve the CFPQ algorithm based on tensor-product proposed by~\cite{10.1007/978-3-030-54832-2_6}.
	We reduce this algorithm to operations over Boolean matrices.
	As a result, all-path query semantics is handled, as opposed to the previous matrix-based solution capable of handling only the single-path semantics.
	Best to our knowledge, our algorithm is the first CFPQ algorithm based on linear algebra which is capable to handle all-path query semantics.
	Also, both regular and context-free grammars can be used as queries.
	\item
	We prove the correctness and time complexity for the proposed algorithm thus providing an upper bound on the complexity of the CFPQ problem in dependence on the size of the query (its context-free grammar) and the number of vertices in the input graph.
	The proposed algorithm has subcubic complexity in terms of the grammar and the input graph sizes, which is comparable with the state-of-the-art solutions.
	On the other hand, the algorithm does not require transforming the input grammar to Chomsky Normal Form.
	The transformation leads to at least a quadratic blow-up in grammar size, thus by avoiding the transformation, our algorithm achieves better time complexity   than other solutions in terms of the grammar size.
	\item We demonstrate the interconnection between CFPQ and incremental transitive closure.
	We show that incremental transitive closure is a bottleneck on the way to achieve a faster CFPQ algorithm for the general case of arbitrary graphs as well as for special families of graphs, such as planar graphs.
	\item We implement the described algorithm and evaluate it on real-world data for both RPQ and CFPQ.
	The evaluation shows that the proposed algorithm is comparable with the existing solutions for CFPQ and RPQ, thus the algorithm provides a promising way to handle both CFPQ and RPQ.
\end{enumerate}
\section{Preliminaries}

In this section, we introduce some basic notation and definitions from graph theory and formal language theory which will be used in the rest of the paper.

\subsection{Language-Constrained Path Querying Problem}

We use a directed edge-labeled graph as a data model.
To introduce the \term{Language-Constraint Path Querying Problem}~\citep{barrett2000formal} over directed edge-labeled graphs we first give definitions for both languages and grammars.

\begin{definition}
An \term{edge-labeled directed graph} $\mathcal{G}$ is a triple $\langle V,E,L \rangle$, where $V = \{0, \ldots, |V|-1\}$ is a finite set of vertices, $E \subseteq V \times L \times V$ is a finite set of edges and $L$ is a finite set of edge labels.
\end{definition}

The graph  $\mathcal{G}$  which we use in the further examples is presented in Figure~\ref{fig:example_input_graph}.

\begin{figure}[h!]
\tikzset{every loop/.style={min distance=10mm,in=0,out=70,looseness=5}}
    \centering
    \begin{tikzpicture}[shorten >=1pt,auto]
       \node[state] (q_0)                      {$0$};
       \node[state] (q_1) [right=of q_0]       {$1$};
        \path[->]
        (q_0) edge[bend left, above]   node {a} (q_1)
        (q_1) edge [bend left, below] node {a} (q_0)
        (q_1) edge[loop right] node {b} (q_1);
    \end{tikzpicture}
    \caption{The example of input graph $\mathcal{G}$}
    \label{fig:example_input_graph}
\end{figure}
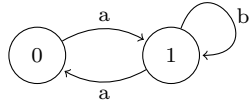

\begin{definition}
An \term{adjacency matrix} for an edge-labeled directed graph $\mathcal{G} = \langle V,E,L \rangle$ is a matrix $M$, where:
\begin{itemize}
    \item $M$ has size $|V|\times|V|$
    \item $M[i,j] = \{l~\mid e = (i,l,j) \in E\}$
\end{itemize}
\end{definition}

Adjacency matrix $M_2$ of the graph $\mathcal{G}$ is

$$
    M_2 =
    \begin{pmatrix}
    \emptyset & \{a\}     \\
   \{a\}   &   \{b\}     \\
    \end{pmatrix}.
$$

\begin{definition}
The \term{Boolean matrices decomposition},
for an edge-labeled directed graph $\mathcal{G} =
\langle V,E,L \rangle$ with adjacency matrix $M$ is a set of matrices $\mathcal{M} = \{ M^l~\mid~l \in L, \ M^l[i,j] = 1 \iff l \in M[i,j]\}$.
\end{definition}

In our work, we use the decomposition of the adjacency matrix into a set of Boolean matrices. As an example, matrix $M_2$ can be represented as a set of two Boolean matrices $M_2^a$ and $M_2^b$.
$$
\mathcal{M}_2 =\bigg\{
M_2^{a} =
\begin{pmatrix}
    . & 1     \\
    1 & .
\end{pmatrix},
M_2^{b} =
\begin{pmatrix}
    . & .    \\
    . & 1
\end{pmatrix}
\bigg\}.
$$

This way we reduce operations necessary for our algorithm from
operations over custom semiring (over edge labels) to operations over a Boolean semiring with an \term{addition} $+$ as a disjunction $\lor$ and a \term{multiplication} $\cdot$ as a conjunction $\land$ over Boolean values.

We also use the notation $\mathcal{M}(\mathcal{G})$ and $\mathcal{G}(\mathcal{M})$ to describe the Boolean decomposition matrices for some graph and the graph formed by its adjacency Boolean matrices.

\begin{definition}
A \term{path} $\pi$ in the graph $\mathcal{G} = \langle V,E,L \rangle$ is a sequence $e_0,e_1,\ldots,e_{n-1}$, where $e_i = (v_i,l_i,u_i) \in E$ and for any $e_i, e_{i+1}$: $u_i = v_{i+1}$. We denote a path from $v$ to $u$ as $v\pi u$.
\end{definition}

\begin{definition}
A \term{word formed by a path} $$\pi = (v_0,l_0,v_1),(v_1,l_1,v_2),\ldots,(v_{n-1},l_{n-1},v_n)$$ is a concatenation of labels along the path: $\omega(\pi) = l_0 l_1 \ldots l_{n-1}$.
\end{definition}

\begin{definition}
A \term{language} $\mathcal{L}$ over a finite alphabet $\Sigma$ is a subset of all possible sequences formed by symbols from the alphabet: $\mathcal{L}_{\Sigma} = \{\omega \mid \omega \in \Sigma^*\}$.
\end{definition}

Now we are ready to introduce language-constraint path querying problem for the given graph  $\mathcal{G} = \langle V,E,L \rangle$ and the given language $\mathcal{L}$ with reachability and all-path semantics.

\begin{definition}
To evaluate language-constraint path query with reachability semantics is to construct a set of pairs of vertices $(v_i,v_j)$ such that there exists a path $v_i \pi v_j$ in $\mathcal{G}$ which forms the word from the given language:
$$
R = \{(v_i,v_j) \mid \exists \pi: v_i \pi v_j, \ \omega(\pi) \in \mathcal{L} \}
$$
\end{definition}

\begin{definition}
To evaluate language-constraint path query with all-path semantics is to construct a set of paths $\pi$ in $\mathcal{G}$ which form the word from the given language:
$$
\Pi = \{ \pi \mid \omega(\pi) \in \mathcal{L}\}
$$
\end{definition}

Note that $\Pi$ can be infinite, thus in practice, we should provide a way to build a finite representation of such paths with reasonable complexity, instead of explicit construction of the $\Pi$.

\subsection{Regular Path Queries and Finite State Machine}

In \term{Regular Path Querying} (RPQ) the language $\mathcal{L}$ is regular.
This case is widespread and well-studied.
The most common way to specify regular languages is by \term{regular expressions}.

We use the following definition of regular expressions.
\begin{definition}
A \term{regular expression} over the alphabet $\Sigma$ is a finite combination of patterns, which can be defined as follows: $\varnothing$ (empty language), $\varepsilon$ (empty string), $a_i \in \Sigma$ are regular expressions, and if $R_1$ and $R_2$ are regular expressions, then $R_1 \mid R_2$ (alternation), $R_1 \cdot R_2$ (concatenation), $R_1^*$ (Kleene star) are also regular expressions.
\end{definition}

For example, one can use regular expression $R_1 = ab^*$ to search for paths in the graph $\mathcal{G}$ (Figure~\ref{fig:example_input_graph}).
The expected query result is a set of paths that start with an $a$-labeled edge and contain zero or more $b$-labeled edges after that.

In this work we use the notion of \term{Finite-State Machine} (FSM) or \term{Finite-State Automaton} (FSA) for RPQs.

\begin{definition}
A \term{deterministic finite-state machine without $\varepsilon$-transitions} $T$ is a tuple $\langle \Sigma, Q, Q_s, Q_f, \delta \rangle$, where:
\begin{itemize}
    \item $\Sigma$ is an input alphabet,
    \item $Q$ is a finite set of states,
    \item $Q_s \subseteq Q$ is a set of start (or initial) states,
    \item $Q_f \subseteq Q$ is a set of final states,
    \item $\delta: Q \times \Sigma \to Q$ is a transition function.
\end{itemize}
\end{definition}

It is well known, that every regular expression can be converted to deterministic FSM without $\varepsilon$-transitions~\citep{automata:theory:10.5555/1177300}.
We use FSM as a representation of RPQ.
FSM $T = \langle \Sigma, Q, Q_s, Q_f, \delta \rangle$ can be naturally represented by a directed edge-labeled graph $\mathcal{G} = \langle V,E,L \rangle$, where $V = Q$, $L = \Sigma$, $E = \{(q_i,l,q_j) \mid \delta(q_i,l) = q_j\}$ and some vertices are marked as the start and final states.
An example of the graph representation of FSM $T_1$ for the regular expression $R_1$ is presented in Figure~\ref{fig:example_fsm}.

\begin{figure}[h]
    \centering
    \begin{tikzpicture}[shorten >=1pt,auto]
       \node[state, initial] (q_0)                      {$0$};
       \node[state, accepting] (q_1) [right=of q_0] {$1$};
       \path[->]
        (q_0) edge  node {a} (q_1)
        (q_1) edge[loop above]  node {b} (q_1)
        ;
    \end{tikzpicture}
    \caption{The example of graph representation of FSM for regular expression $ab^*$}
    \label{fig:example_fsm}
\end{figure}
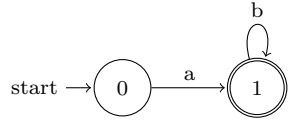

As a result, FSM also can be represented as a set of Boolean adjacency matrices $\mathcal{M}$ accompanied by the information about the start and final vertices.
For example, FSM $T_1$ can be represented as follows.
$$
M_1^a =
\begin{pmatrix}
.&1 \\
.&.
\end{pmatrix},~
M_1^b =
\begin{pmatrix}
.&. \\
.&1
\end{pmatrix}.
$$

Note, that an edge-labeled graph can be viewed as an FSM where edges represent transitions and all vertices are both start and final at the same time.
Thus RPQ evaluation is an intersection of two FSMs.
The query result can also be represented as FSM because regular languages are closed under intersection~\citep{automata:theory:10.5555/1177300}.

\subsection{Context-Free Path Querying and Recursive State Machines}

An even more general case than RPQ is a \term{Context-Free Path Querying Problem (CFPQ)}, where one can use context-free languages as constraints.
These constraints are more expressive than regular ones.
For example, a classic same-generation query can be expressed by a context-free language, but not a regular language \citep{databasebook}.

\begin{definition}
A \term{context-free grammar} $G$ is a tuple $\langle\Sigma, N, S, P\rangle$, where:
\begin{itemize}
    \item $\Sigma$ is a finite set of terminals (or terminal alphabet)
    \item $N$ is a finite set of nonterminals (or nonterminal alphabet)
    \item $S \in N$ is a start nonterminal
    \item $P$ is a finite set of productions (grammar rules) of form $N_i \to \alpha$ where  $N_i \in N$, $\alpha \in (\Sigma \cup N)^*$.
\end{itemize}
\end{definition}

\begin{definition}

    The \emph{size of the grammar} $G = \langle \Sigma, N, S, P \rangle$ is defined as the sum of the sizes of its productions:
    \[|G| = \sum_{p \in P} |p|,\]
    where the \term{size of a production} $p = N \to \alpha$ depends on the length of its right-hand side:
    \[|p| = 1 + |\alpha|.\]

\end{definition}


\begin{definition}
The sequence $\omega_2 \in (\Sigma \cup N)^*$ is \term{derivable} from $\omega_1 \in (\Sigma \cup N)^*$ \term{in one derivation step}, or $\omega_1 \to \omega_2$, in the grammar $G = \langle\Sigma, N, S, P\rangle$ iff $\omega_1=\alpha N_i \beta$, $\omega_2 = \alpha \gamma \beta$, and $N_i \to \gamma \in P$.
\end{definition}

\begin{definition}
Context-free grammar $G=\langle\Sigma, N, S, P\rangle$ specifies a \term{context-free language}: $\mathcal{L}(G) = \{\omega \mid S \xrightarrow{*} \omega \}$, where $(\xrightarrow{*})$ denotes zero or more derivation steps $(\to)$.
\end{definition}

For instance, the grammar $G_1 = \langle \{a,b\}, \{S\}, S, \{S \to a \ b; \ S \to a \ S \ b\} \rangle$ can be used to search for paths, which form words in the language $\mathcal{L}(G_1) = \{a^n b^n \mid n > 0\}$ in the graph $\mathcal{G}$ (Figure~\ref{fig:example_input_graph}).


While a regular expression can be transformed to an FSM, a context-free grammar can be transformed to a \term{Recursive State Machine} (RSM) in a similar fashion.
In our work, we use the following definition of RSM based
on~\cite{rsm:analysis:10.1007/3-540-44585-4_18}.

\begin{definition}
A \term{recursive state machine} $R$ over a finite alphabet $\Sigma$ is defined as a tuple of elements $\langle B,m,\{C_i\}_{i \in B} \rangle$, where:

\begin{itemize}
    \item $B$ is a finite set of labels of boxes,
    \item $m \in B$ is an initial box label,
    \item Set of \term{component state machines} or \term{boxes},
          where $C_i=(\Sigma \cup B, Q_i,q_i^0,F_i,\delta_i)$:
    \begin{itemize}
        \item $\Sigma \cup B$ is a set of symbols, $\Sigma \cap B = \varnothing$,
        \item $Q_i$ is a finite set of states,
              where $Q_i \cap Q_j =  \varnothing, \forall i \neq j$,
        \item $q_i^0$ is an initial state for $C_i$,
        \item $F_i$ is a set of final states for $C_i$, where $F_i \subseteq Q_i$,
        \item $\delta_i: Q_i \times (\Sigma \cup B) \to Q_i$ is a transition function. 
    \end{itemize}
\end{itemize}

\end{definition}

\begin{definition}
    The \term{size of RSM} $|R|$ is defined as the sum of the number of states in all boxes.
\end{definition}

RSM behaves as a set of finite state machines (or FSM).
Each such FSM is called a \term{box} or a \term{component state machine}.
A box works similarly to the classic FSM, but it also handles additional \term{recursive calls} and employs an implicit \term{call stack} to \term{call} one component from another and then return execution flow back.

The execution of an RSM could be defined as a sequence of the configuration transitions, which are done while reading the input symbols.
The pair $(q_i,\mathcal{S})$, where $q_i$ is a current state for box $C_i$ and $\mathcal{S}$ is  a stack of \term{return states}, describes an \term{execution configuration}.

The RSM execution starts from the configuration $(q_m^0, \langle\rangle)$.
The following list of rules defines the machine transition from configuration $(q_i,\mathcal{S})$ to $(q',\mathcal{S}')$ on some input symbol $a$:

\begin{itemize}
    \item $(q_i^k,\mathcal{S}) \leadsto (\delta_i (q_i^k, a),\mathcal{S})$
    \item $(q_i^k,\mathcal{S}) \leadsto (q_j^0, \delta_i (q_i^k, j) \circ\mathcal{S})$
    \item $(q_j^k,q_i^t\circ \mathcal{S}) \leadsto (q_i^t, \mathcal{S}),$ where $q_j^k \in F_j$
\end{itemize}

An input word $a_1 \dots a_n$ is accepted, if machine reaches configuration $(q,\langle\rangle)$, where $q \in F_m$.
Note, that an RSM makes nondeterministic transitions and does not read the input character when it \term{calls} some component or \term{returns}.

According to~\cite{rsm:analysis:10.1007/3-540-44585-4_18}, recursive state machines are equivalent to pushdown systems.
Since pushdown systems are capable of accepting context-free languages~\citep{automata:theory:10.5555/1177300}, RSMs are equivalent to context-free languages.
Thus RSMs are suitable to encode query grammars.
Any CFG $G=\langle\Sigma, N, S, P\rangle$ can be easily converted to an RSM $R_G$ with one box per nonterminal.
The box which corresponds to a nonterminal $N_i$ is constructed using the right-hand side of each rule for $N_i$.
Such conversion implies that $|R_G| = O(|G|) = O(|P|)$ because the number of states in $R_G$ does not exceed the sum of the lengths of every rule in $G$ in the worst case.

An example of such RSM $R$ constructed for the grammar $G_1$ with rules $S \to a S b \mid a b$ is provided in Figure~\ref{example:automata}.
For the given example of the grammar and the RSM, consider the following sequence of the machine configuration transitions, in the case where one wants to determine if the input word $aabb$ belongs to the language $L(G_1)$.
The RSM execution starts from configuration $(q_S^0,\langle \rangle)$, reads the symbol $a$ and goes the to $(q_S^1, \langle \rangle)$.
Then, in a nondeterministic manner it tries to read $b$ but fails, and in the same time tries to derive $S$ and goes to the configuration $(q_S^0, \langle q_S^2 \rangle)$, where $q_S^2$ is a \textit{return} state.
Then machine reads $a$ and goes to $(q_S^1, \langle q_S^2 \rangle)$. In this case, it fails to derive $S$ in the nondeterministic choice, but successfully reads $b$ and goes to the configuration $(q_S^3,\langle q_S^2 \rangle)$.
Since $q_S^3$ is a final state for the box $S$, the RSM tries to $return$ and goes to $(q_S^2,\langle \rangle)$.
Then it reads $b$ and transits to $(q_S^3,\langle \rangle)$.
Since $q_S^3 \in F_S$ and the stack of return states is empty, the machine accepts the input sequence $aabb$.
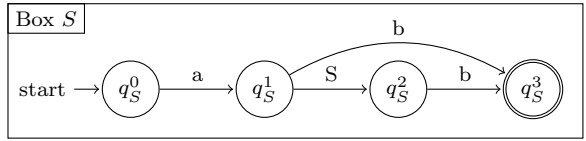
\begin{figure}[h]
    \begin{tikzpicture}[shorten >=1pt,auto]
        \node[state, initial] (q_0)   {$q_S^0$};
        \node[state] (q_1) [right=of q_0] {$q_S^1$};
        \node[state] (q_2) [right=of q_1] {$q_S^2$};
        \node[state, accepting] (q_3) [right=of q_2] {$q_S^3$};
        \path[->]
            (q_0) edge node {a} (q_1)
            (q_1) edge node {S} (q_2)
            (q_2) edge node {b} (q_3)
            (q_1) edge [bend left, above]  node {b} (q_3);
        \node (box) [draw=black, fit= (q_0) (q_1) (q_2) (q_3), xshift=-4ex, yshift=2ex, inner xsep=0.75cm, inner ysep=0.5cm] {};

        \node[draw=black, anchor=north west] at (box.north west) {Box $S$};
    \end{tikzpicture}
    \centering
    \caption{The recursive state machine $R$ for grammar $G_1$}
    \label{example:automata}
\end{figure}

Similarly to an FSM, an RSM can be represented as a graph and, hence, as a set of Boolean adjacency matrices.
For our example, $M_1$ for the RSM $R$ from Figure~\ref{example:automata} is:
    $$
    M_1 =
    \begin{pmatrix}
    \emptyset & \{a\} &   \emptyset &  \emptyset     \\
     \emptyset &  \emptyset & \{S\} & \{b\} \\
     \emptyset &  \emptyset &  \emptyset & \{b\}     \\
     \emptyset &  \emptyset &   \emptyset&   \emptyset
    \end{pmatrix}
    $$

Matrix $M_1$ can be represented as a set of Boolean matrices as follows:
{\small
\begin{align*}
M_1^S =
\begin{pmatrix}
    . & . & . & .   \\
    . & . & 1 & .   \\
    . & . & . & .   \\
    . & . & . & .
\end{pmatrix},~M_1^a =
\begin{pmatrix}
   . & 1 & . & .   \\
   . & . & . & .   \\
   . & . & . & .   \\
   . & . & . & .   \\
\end{pmatrix},~M_1^b =
\begin{pmatrix}
    . & . & . & .   \\
    . & . & . & 1   \\
    . & . & . & 1   \\
    . & . & . & .   \\
\end{pmatrix}
\end{align*}
}
Similarly to an RPQ, a CFPQ is the intersection of the given context-free language and an FSM specified by the given graph.
As far as every context-free language is closed under the intersection with regular languages~\citep{automata:theory:10.5555/1177300}, such intersection can be represented as an RSM.
Also, an RSM can be viewed as an FSM over $\Sigma \cup N$.
In this work, we use this point of view to propose a unified algorithm to evaluate both regular and context-free path queries with zero overhead for regular queries.

\subsection{Graph Kronecker Product and Machines Intersection}

In this section, we introduce the classic Kronecker product definition,
describe graph Kronecker product and its relation to Boolean matrices algebra,
and RSM and FSM intersection.

\begin{definition}
Given two matrices $A$ and $B$ of sizes $m_1 \times n_1$ and $m_2 \times n_2$
respectively, with element-wise product operation $\cdot$, the \term{Kronecker product} of these two matrices is a new matrix $C = A \otimes B$ of size $m_1 * m_2 \times n_1 * n_2$ and \[C[u * m_2 + v,n_2 * p + q] = A[u,p] \cdot B[v,q].\]
\end{definition}

It is worth to mention, that the Kronecker product produces blocked matrix $C$,
with total number of the blocks $m_1 * n_1$, where each block has size
$m_2 * n_2$ and is defined as $A[i,j] \cdot B$.

\begin{definition}
\label{def:graph:product}
Given two edge-labeled directed graphs $\mathcal{G}_1=\langle V_1, E_1, L_1 \rangle$
and $\mathcal{G}_2=\langle V_2, E_2, L_2 \rangle$,
the \term{Kronecker product} of these two graphs is a edge-labeled directed graph
$\mathcal{G}=\mathcal{G}_1 \otimes \mathcal{G}_2$,
where $\mathcal{G}= \langle V, E, L \rangle$:
\begin{itemize}
    \item $V = V_1 \times V_2$
    \item $E = \{((u,v),l,(p,q)) \mid (u,l,p) \in E_1 \wedge (v,l,q) \in E_2 \}$
    \item $L = L_1 \cap L_2$
\end{itemize}
\end{definition}

The Kronecker product for graphs produces a new graph with a property
that if and only if some path $(u,v)\pi(p,q)$ exists in the result graph
then paths $u\pi_1p$ and $v\pi_2q$ exist in the input graphs,
and $\omega(\pi) = \omega(\pi_1) = \omega(\pi_2)$.
These paths $\pi_1$ and $\pi_2$ can easily be found from $\pi$ by its definition.

The Kronecker product for directed graphs can be described as
the Kronecker product of the corresponding adjacency matrices of graphs,
what gives the following definition:

\begin{definition}
Given two adjacency matrices $M_1$ and $M_2$ of sizes
$m_1 \times n_1$ and $m_2 \times n_2$ respectively
for some directed graphs $\mathcal{G}_1$ and $\mathcal{G}_2$,
the \term{Kronecker product} of these two adjacency matrices is the adjacency matrix $M$
of some graph $\mathcal{G}$, where $M$ has size $m_1 * m_2 \times n_1 * n_2$ and
\[M[u * m_2 + v,n_2 * p + q] = M_1[u,p] \cap M_2[v,q].\]
\end{definition}

By definition, the Kronecker product for adjacency matrices gives an
adjacency matrix with the same set of edges as in the resulting graph in the
Def.~\ref{def:graph:product}. Thus, $M(\mathcal{G}) = M(\mathcal{G}_1) \otimes
M(\mathcal{G}_2)$, where $\mathcal{G} = \mathcal{G}_1 \otimes \mathcal{G}_2$.

\begin{definition}
\label{def:fsm:intersection}
Given two finite state machines $T_1 = \langle \Sigma, Q^1, Q_S^1, Q_F^1, \delta^1 \rangle$ and $T_2 = \langle \Sigma, Q^2, Q_S^2, Q_F^2, \delta^2 \rangle$, the \term{intersection} of these two machines is a new FSM $T = \langle \Sigma, Q, Q_S, Q_F, \delta \rangle$, where:
\begin{itemize}
    \item $Q = Q^1 \times Q^2$
    \item $Q_S = Q_S^1 \times Q_S^2$
    \item $Q_F = Q_F^1 \times Q_F^2$
    \item $\delta: Q \times \Sigma \to Q$,
    $\delta (\langle q_1, q_2 \rangle, s) = \langle q_1', q_2' \rangle$, if $\delta(q_1,s)=q_1'$ and $\delta(q_2,s)=q_2'$
\end{itemize}
\end{definition}

According to~\cite{automata:theory:10.5555/1177300} an FSM intersection defines the machine for which $L(T) = L(T_1) \cap L(T_2)$.

The most substantial part of the intersection is the $\delta$ function construction for the new machine $T$.
Using adjacency matrices representation for FSMs, we can reduce the intersection to the Kronecker product of such matrices over Boolean semiring to some extent, since the transition function $\delta$ of the machine $T$ in the matrix form is exactly the same as the product result.
More precisely:

\begin{definition}
\label{def:bool:product}
Given two sets of Boolean adjacency matrices $\mathcal{M}_1$ and $\mathcal{M}_2$, the \term{Kronecker product} of these matrices is a new matrix
$\mathcal{M} = \mathcal{M}_1 \otimes \mathcal{M}_2$, where $\mathcal{M} = \{ M_1^a \otimes M_2^a~|~a \in \Sigma \}$ and the element-wise operation is a conjunction over Boolean values ($\wedge$).
\end{definition}

Applying the Kronecker product for both the FSM and the edge-labeled directed graph, we can intersect these objects as shown in Def.~\ref{def:bool:product}, since the graph could be interpreted as an FSM with transitions matrix represented as the Boolean adjacency matrix.

In this work, we show how to express RSM and FSM intersection in terms of
the Kronecker product and transitive closure over a Boolean semiring.
\section{Context-free path querying by Kronecker product}

In this section, we introduce the algorithm for context-free path querying which is based on the Kronecker product of Boolean matrices.
The algorithm solves all-pairs CFPQ problem with all-path semantics (according to~\cite{hellingsPathQuerying}).
The algorithm works in the following two steps.
\begin{enumerate}
\item \emph{Index creation}.
 During this step, the algorithm computes an index that contains the information necessary to restore paths for the given pairs of vertices.
 This index can be used to solve the reachability problem without extracting paths.
 Note that the index is finite even when the set of paths is infinite.
\item \emph{Paths extraction}.
All paths for the given pair of vertices can be enumerated by using the index.
Since the set of paths can be infinite, all paths cannot be enumerated explicitly, thus advanced techniques such as lazy evaluation are required for the implementation.
Nevertheless, a single path can always be extracted with standard techniques.
\end{enumerate}

In the following subsections, we describe these steps, prove the correctness of the algorithm, and provide time complexity estimations.
For the first step, we start by introducing a naive algorithm.
After that, we show how to achieve cubic time complexity by using a dynamic transitive closure algorithm and shave off a logarithmic factor to achieve the best known time complexity for the CFPQ problem.
We finish by providing a step-by-step example of query evaluation with the proposed algorithm.

\subsection{Index Creation Algorithm}

The \textit{index creation} algorithm outputs the final adjacency matrix for the input graph with all pairs of vertices which are reachable through some nonterminal in the input grammar $G$, as well as the index matrix which is to be used to extract paths in the \textit{path extraction} algorithm.

The algorithm is based on the generalization of the FSM intersection for an RSM,  and the edge-labeled directed input graph.
Since the RSM is composed as a set of FSMs, it could easily be presented as an adjacency matrix for some graph over the set of labels.
As shown in the Def.~\ref{def:bool:product}, we can apply Kronecker product for Boolean matrices to \textit{intersect} the RSM and the input graph to some extent.
But the RSM contains nonterminal symbols with the additional logic of \textit{recursive calls}, which requires a \textit{transitive closure} step to extract such symbols.

The core idea of the algorithm comes from the Kronecker product and transitive closure.
The algorithm boils down to the evaluation of the iterative Kronecker product for the adjacency matrix $\mathcal{M}_1$ of the RSM $R$ and the adjacency matrix $\mathcal{M}_2$ of the  input graph $\mathcal{G}$, followed by the transitive closure, extraction of nonterminals and updating the graph adjacency matrix $\mathcal{M}_2$.
Listing~\ref{tensor:cfpq:cubic} demonstrates the main steps of the algorithm.
\begin{algorithm}[h]
\floatname{algorithm}{Listing}
\begin{algorithmic}[1]
\footnotesize
\caption{Kronecker product based CFPQ using dynamic transitive closure}
\label{tensor:cfpq:cubic}
\Function{contextFreePathQuerying}{G, $\mathcal{G}$}
   \State{$n \gets$ the number of nodes in $\mathcal{G}$}
    \State{$R \gets$ recursive automata for $G$}
    \State{$\mathcal{M}_1 \gets$ the set of adjacency matrices for $R$}
    \State{$\mathcal{M}_2, \mathcal{A}_2 \gets$ the sets of adjacency matrices for $\mathcal{G}$}
    \State{$C_3, M_3 \gets$ the empty matrices of size $dim(\mathcal{M}_1)n \times dim(\mathcal{M}_1)n$}
    \For{$s \in 0..dim(\mathcal{M}_1)-1$}
        \For{$S \in \textit{getNonterminals}(R,s,s)$}
            \For{$i \in 0..dim(\mathcal{M}_2)-1$}
                \State{$M^S_2[i,i] \gets 1 $}
            \EndFor
        \EndFor
    \EndFor
    \While{$\mathcal{M}_2$ is changing}
        \State{$M_3' \gets \bigvee_{M^S \in \mathcal{M}_1 \otimes \mathcal{A}_2} M^S$}
        \Comment{Using only new edges from $\mathcal{A}_2$}
        \State{$M_3 \gets M_3 + M_3'$}
        \Comment{Updating the matrix for the Kronecker product result}
        \State{$\mathcal{A}_2 \gets$ The empty matrix}
        \State{$C_3' \gets$ The empty matrix of size $dim(\mathcal{M}_1)n \times dim(\mathcal{M}_1)n$}
        \For{$(i,j) \mid M_3'[i,j] \neq 0$}
            \State{$C_3' \gets \textit{add}(C_3, C_3', i, j)$}
            \Comment{Updating the transitive closure}
            \State{$C_3 \gets C_3 + C_3'$}
        \EndFor
        \For{$(i,j)\ |\ C_3'[i,j] \neq 0$}
                \State{$s, f \gets \textit{getStates}(C_3',i,j)$}
                \State{$x, y \gets \textit{getCoordinates}(C_3',i,j)$}
                \For{$S \in \textit{getNonterminals}(R,s,f)$}
                    \State{$M^S_2[x,y] \gets 1$}
                    \State{$A^S_2[x,y] \gets 1$}
                \EndFor
        \EndFor
    \EndWhile
\State \Return $\mathcal{M}_2, M_3$
\EndFunction
\Function{getStates}{$C, i, j$}
    \State{$n \gets dim(\mathcal{M}_2)$}
   \Comment{$\mathcal{M}_2$ the set of adjacency matrices for $\mathcal{G}$}
    \State \Return{$\left\lfloor{i / n}\right\rfloor, \left\lfloor{j / n}\right\rfloor$}
\EndFunction
\Function{getCoordinates}{$C, i, j$}
    \State{$n \gets dim(\mathcal{M}_2)$}
    \State \Return{$i \bmod n, j \bmod n$}
\EndFunction
\end{algorithmic}
\end{algorithm}
\subsubsection{Application of Dynamic Transitive Closure}
The most time-consuming steps of the algorithm are the computations of the Kronecker product and transitive closure.
Note that the adjacency matrix $\mathcal{M}_2$ is changed incrementally i.e. elements (edges) are added to $\mathcal{M}_2$ at each iteration of the algorithm and are never deleted from it.
So it is not necessary to recompute the whole product or transitive closure if some appropriate data structure is maintained.

To compute the Kronecker product, we employ the fact that it is left-distributive.
Let $\mathcal{A}_2$ be a matrix with newly added elements and $\mathcal{B}_2$ be a matrix with all previously found elements, such that $\mathcal{M}_2 = \mathcal{A}_2 + \mathcal{B}_2$.
Then by left-distributivity of the Kronecker product we have $\mathcal{M}_1 \otimes \mathcal{M}_2 = \mathcal{M}_1 \otimes (\mathcal{A}_2 + \mathcal{B}_2) = \mathcal{M}_1\otimes \mathcal{A}_2 + \mathcal{M}_1 \otimes \mathcal{B}_2$.
Note that $\mathcal{M}_1 \otimes \mathcal{B}_2$ is known and is already in the matrix $\mathcal{M}_3$ and its transitive closure is also already in the matrix $C_3$, because it has been calculated at the previous iterations, so it is left to update some elements of $\mathcal{M}_3$ by computing $\mathcal{M}_1\otimes \mathcal{A}_2$.

The fast computation of transitive closure can be obtained by using an incremental dynamic transitive closure technique.
Let us describe the function $add$ from Listing \ref{tensor:cfpq:cubic}.
Let $C_3$ be a transitive closure matrix of the graph $G$ with $n$ vertices.
We use an approach by~\cite{IBARAKI198395} to maintain dynamic transitive closure.
The key idea of their algorithm is to recalculate reachability information only for those vertices which become reachable after insertion of a certain edge.

We have modified the algorithm to achieve a logarithmic speed-up in the following way.
For each newly inserted edge $(i, j)$ and every node $u \neq j$ of $G$ such that $C_3[u, i] = 1$ and $C_3[u, j]=0$, one needs to perform operation $C_3[u,v] = C_3[u, v] \wedge C_3[j, v]$ for every node $v$, where $1 \wedge 1 = 0 \wedge 0 = 1 \wedge 0 = 0$ and $0 \wedge 1 = 1$.
Notice that these operations are equivalent to the element-wise (Hadamard) product of two vectors of size $n$, where multiplication operation is denoted as $\wedge$. To check whether $C_3[u, i] = 1$ and $C_3[u, j]=0$ one needs to multiply two vectors: the first vector represents reachability of the given vertex $i$ from other vertices $\{u_1, u_2, ..., u_n\}$ of the graph and the second vector represents the same for the given vertex $j$. The operation $C_3[u, v] \wedge C_3[j, v]$ also can be reduced to the computation of the Hadamard product of two vectors of size $n$ for the given $u_k$. The first vector contains the information whether vertices  $\{v_1, v_2, ..., v_n\}$ of the graph are reachable from the given vertex $u_k$ and the second vector represents the same for the given vertex $j$. The element-wise product of two vectors can be calculated naively in time $O(n)$. Thus, the time complexity of the transitive closure can be reduced by speeding up the element-wise product of two vectors of size $n$.

To achieve logarithmic speed-up, we use the Four Russians' trick by \cite{arlazarov1970economical}.
Let us assume an architecture with word size $w= \theta(\log n)$.
First we split each vector into $n/\log n$ parts of size $\log n$.
Then we create a table $T$ such that $T(a, b)$ = $a \wedge b$ where $a, b \ \in {\{0,1\}}^{\log n}$.
This takes time $O(n^2 \log n)$, since there are $2^{\log n} = n$ variants of Boolean vectors of size $\log n$ and hence $n^2$ possible pairs of vectors $(a, b)$ in total, and each component takes $O(\log n)$ time.
Assuming constant-time logical operations on words, we can store a polynomial number of lookup tables (arrays) $T_i$ (one array for each vector of size $\log n$), such that given an index of a table $T_i$, and any $O(\log n)$ bit vector $b$, we can look up $T_i(b)$ in constant time. The index of each array $T_a$ is stored in array $T$, which can be accessed in constant time for a given $\log$-size vector $a$. Thus, we can calculate the product of two parts $a$ and $b$ of size $\log n$ in constant time using the table $T$.
There are $n/\log n$ such parts, so the element-wise product of two vectors of size $n$ can be calculated in time $O(n/\log n)$ with $O(n^2 \log n)$ preprocessing.

\begin{theorem}
    Let $\mathcal{G} =  \langle V,E,L\rangle$ be a graph and $G = \langle\Sigma, N, S, P\rangle$ be a grammar.
    Let $\mathcal{M}_{2}$ be a resulting adjacency matrix after the execution of the algorithm in Listing~\ref{tensor:cfpq:cubic}. Then for any valid indices $i, j$ and for each nonterminal $N_i \in N$ the following statement holds: the cell $M_{2,(k)}^{N_i}[i,j]$ contains $\{1\}$, iff there is a path $i\pi j$ in the graph $\mathcal{G}$ such that $ N_i \xrightarrow{*} l(\pi)$.
\end{theorem}{}
\begin{proof}
    By induction on the height of the derivation tree obtained on each iteration.
\end{proof}{}

\begin{theorem}{}
\label{theorem: subcubic}
    Let $\mathcal{G} = \langle V,E,L \rangle$ be a graph and $G = \langle\Sigma, N, S, P\rangle$ be a grammar.
    The algorithm from Listing~\ref{tensor:cfpq:cubic} calculates resulting matrices $\mathcal{M}_2$ and $M_3$ in $O({|P|}^3n^3/\log (|P|n))$ time on a word RAM with word size $w= \theta(\log |P|n)$, where $n = |V|$. Moreover, maintaining of the dynamic transitive closure dominates the cost of the algorithm.
\end{theorem}

\begin{proof}
 Let $|\mathcal{A}|$ be the number of non-zero elements in a matrix $\mathcal{A}$. Consider the total time which is needed for computing the Kronecker products. The elements of the matrices $\mathcal{A}_2^{(i)}$ are pairwise distinct on every $i$-th iteration of the algorithm therefore the total number of operations is $\sum\limits_i{T(\mathcal{M}_1 \otimes \mathcal{A}_2^{(i)})} = |\mathcal{M}_1| \sum\limits_i {|\mathcal{A}_2^{(i)}|} = (|N| + |\Sigma|){|P|}^2 \sum\limits_i {|\mathcal{A}_2^{(i)}|} = O({(|N| + |\Sigma|)}^2{|P|}^2 n^2).$

Now we derive the time complexity of maintaining the dynamic transitive closure.
Notice that $C_3$ has a size of the Kronecker product of $\mathcal{M}_1 \otimes \mathcal{M}_2$, which is equal to $dim(\mathcal{M}_1)n \times dim(\mathcal{M}_1)n = |P|n \times |P|n$ so no more than ${|P|}^2n^2$ edges will be added during all iterations of the Algorithm.
Checking whether $C_3[u, i] = 1$ and $C_3[u, j]=0$ for every node $u \in V$ for each newly inserted edge $(i, j)$ requires one multiplication of vectors per insertion, thus total time is $O({|P|}^3n^3/\log (|P|n))$.
Note that after checking the condition, at least one element $C[u', j]$ changes value from 0 to 1 and then never becomes 0 for some $u'$ and $j$.
Therefore the operation $C_3[u',v] = C_3[u', v] \wedge C_3[j, v]$ for all $v \in V$ is executed at most once for every pair of vertices $(u',j)$ during the entire computation implying that the total time is equal to $O({|P|}^2n^2|P|n/\log (|P|n))=O({|P|}^3n^3/\log (|P|n))$, using the  multiplication of vectors.

The matrix $C_3'$ contains only new elements, therefore $C_3$ can be updated directly using only $|C_3'|$ operations and hence ${|P|}^2n^2$ operations in total.
The same holds for the loop in line 19 of the algorithm from Listing~\ref{tensor:cfpq:cubic}, because operations are performed only for non-zero elements of the matrix $|C_3'|$.
Finally, the time complexity of the algorithm is $O({(|N| + |\Sigma|)}^2{|P|}^2 n^2) + O({|P|}^2n^2) + O({|P|}^2n^2 \log (|P|n)) + O({|P|}^3n^3/\log (|P|n)) + O({|P|}^2n^2)= O({|P|}^3n^3/\log (|P|n))$. \qed
\end{proof}

The complexity analysis of the Algorithm~\ref{tensor:cfpq:cubic} shows that the maintaining of the incremental transitive closure dominates the cost of the algorithm. Thus, CFPQ can be solved in truly subcubic $O(n^{3-\varepsilon})$ time if there exists an incremental dynamic algorithm for the transitive closure for a graph with $n$ vertices with preprocessing time $O(n^{3-\varepsilon})$ and total update time $O(n^{3-\varepsilon})$. Unfortunately, such an algorithm is unlikely to exist: it was proven by~\cite{10.1145/2746539.2746609} that there is no incremental dynamic transitive closure algorithm for a graph with $n$ vertices and at most $m$ edges with preprocessing time $poly(m)$, total update time $mn^{1-\varepsilon}$, and query time $m^{\delta-\varepsilon}$ for any $\delta \in (0, 1/2]$ per query that has an error probability of at most 1/3 assuming the widely believed Online Boolean Matrix-Vector Multiplication (OMv) Conjecture. OMv Conjecture introduced by~\cite{10.1145/2746539.2746609} states that for any constant $ \varepsilon>0$, there is no $O(n^{3-\varepsilon})$-time algorithm that solves OMv with an error probability of at most 1/3.



\subsubsection{Index creation for RPQ}
In the case of the RPQ, the main \textbf{while} loop takes only one iteration to actually append data.
Since the input query is provided in the form of the regular expression, one can construct the corresponding RSM which consists of the single \textit{component state machine}.
This CSM is built from the regular expression and is labeled as $S$, for example, and has no \textit{recursive calls}.
The~adjacency matrix of the machine is built over $\Sigma$ only.
Therefore, during the calculation of the Kronecker product, all relevant information is taken into account at the first iteration of the loop.
\subsection{Paths Extraction Algorithm}
After the index has been created, one can enumerate all paths between specified vertices.
The index $M_3$ already stores data about all paths derivable from nonterminals.
This data can be used to construct these paths. However, the set of such paths can be infinite.
From a practical perspective, it is necessary to use lazy evaluation or limit the resulting set of paths in some other way.
For example, one can try to query some fixed number of paths, query paths of fixed maximum length, or just query a single path.
The problem of paths enumeration, namely how to enumerate paths with the smallest delay, may be relevant here.

\begin{algorithm}[h]
	\floatname{algorithm}{Listing}
	\begin{algorithmic}[1]
		\footnotesize
		\caption{Paths extraction algorithm}
		\label{tensor:pathsExtraction}
		\State{$M_3 \gets $ the result of index creation algorithm: final Kronecker product}
		\State{$R \gets$ recursive automata for the input RSM}
		\State{$\mathcal{M}_1 \gets $  the set of adjacency matrices for $R$}
		\State{$\mathcal{M}_2 \gets $ the set of adjacency matrices of the final graph}

		\Function{getPaths}{$v_s, v_f, N$}
		\State{$q_N^0 \gets$ Start state of automata for $N$}
		\State{$F_N \gets$ Final states of automata for $N$}
		\State{$paths\gets \bigcup\limits_{q_N^f \in F_N} \Call{genPaths}{(q_N^0,v_s),(q_N^f, v_f)}$}
		\State{$resultPaths \gets \emptyset$}
		\For{$path \in paths$}
		\State{$currentPaths \gets \emptyset$}
		\For{$((s_i, v_i), (s_j, v_j)) \in path$}
		\State{\begin{minipage}[t]{0.2\textwidth}
				\vspace{-13pt}
				\begin{align*}
					currentSubPaths \gets & \{(v_i,t,v_j) \mid M_2^t[v_i, v_j] \wedge M_1^t[s_i, s_j]\}\\
					& \cup \ \bigcup_{\{N \mid  M_2^N[v_i, v_j] \wedge M_1^N[s_i, s_j]\}}\Call{getPaths}{v_i, v_j, N}
				\end{align*}
			\end{minipage}
		}
		\State{$currentPaths \gets currentPaths \cdot currentSubPaths$}
		\Comment{Concatenation of paths}
		\EndFor
		\State{$resultPaths \gets resultPaths \cup currentPaths$}
		\EndFor

		\State \Return $resultPaths$
		\EndFunction


		\Function{genPaths}{$(s_i,v_i), (s_j,v_j)$}
		\State{$q \gets \text{vector of zeros with size } dim(M_3)$}
		\Comment{Vector for indicating the current vertex}
		\State{$q[s_i dim(\mathcal{M}_2) + v_i] \gets 1$}
		\State{$resultPaths \gets \emptyset$}
		\State{$supposedPaths \gets \{([~], q)\}$}
		\Comment{Set of pairs: path and the vector for current vertex}
		\While{$supposedPaths$ is changing}
		\For{$(path, q) \in supposedPaths$}
		\If{$q[s_j dim(\mathcal{M}_2) + v_j] = 1$}
		\State{$resultPaths \gets resultPaths \cup path$}
		\EndIf

		\State{$q \gets q \cdot (M_3)^T$}
		\Comment{Boolean vector-matrix multiplication}

		\State{$\text{Remove } (path, q)\text{ from } supposedPaths$}
		\For{$j \text{ such that } q[j] = 1$}
		\State{$pathNew \gets path$}
		\If{$path \text{ is empty path}$}
		\State{$pathNew \gets pathNew \cdot [\big((s_i, v_i), (\left\lfloor{j / dim(\mathcal{M}_2)}\right\rfloor, j \bmod dim(\mathcal{M}_2))\big)]$}
		\Else
		\State{$(s_k, v_k) \gets \text{ the last vertex of } path$}
		\State{$pathNew \gets pathNew \cdot [\big((s_k, v_k), (\left\lfloor{j / dim(\mathcal{M}_2)}\right\rfloor, j \bmod dim(\mathcal{M}_2))\big)]$}
		\EndIf
		\State{$qNew \gets \text{vector of zeros with size } dim(M_3)$}
		\State{$qNew[j] \gets 1$}
		\State{$\text{Add } (pathNew, qNew)\text{ to } supposedPaths$}
		\EndFor
		\EndFor
		\EndWhile
		\State \Return $resultPaths$
		\EndFunction
	\end{algorithmic}
\end{algorithm}

The most natural way to use the created index is to query paths between the specified vertices derivable from the specified nonterminal. To do so, we provide a function \textsc{getPaths}($v_s, v_f, N$), where $v_s$ is a start vertex of the graph, $v_f$ --- the final vertex, and $N$ is a nonterminal.
Implementation of this function is presented in Listing~\ref{tensor:pathsExtraction}.

Paths extraction is implemented as two functions.
The entry point is \textsc{getPaths}($v_s, v_f, N$).
This function returns a set of the paths in the input graph between $v_s$ and $v_f$ such that the word formed by a path is derivable from the nonterminal $N$.

To compute such paths, it is necessary to compute paths from vertices of the form $(q_N^0,v_s)$ to vertices of the form $(q_N^f, v_f)$ in the resulting Kronecker product $M_3$, where $q_N^0$ is an initial state of RSM for $N$ and $q_N^f$ is one of the final states.
For this reason, we provide the function \textsc{genPaths}$((s_i,v_i),(s_j,v_j))$. This function explores the graph corresponding to the resulting Kronecker product $M_3$ in a breadth-first manner and return all paths from $(s_i,v_i)$ to $(s_j,v_j)$. For each path, we also store the vector $q$ that indicates which vertex is current now. We find the next vertices using the Boolean vector-matrix multiplication in line 26 of the algorithm from Listing~\ref{tensor:pathsExtraction}. After that, we add new edges to our paths in lines 31 and 34. If we reach the vertex $(s_j,v_j)$ then we can add the collected path to the resulting set (see lines 24-25). The paths constructed by \textsc{genPaths} is used to construct the corresponding paths in the input graph. Each edge $((s_i,v_i),(s_j,v_j))$ corresponds to set of paths in the input graph. This set is computed in line 13 and is used as subpaths for constructing the resulting paths in line 14. Note that in lines 14, 31, and 34 we use the operation $\cdot$
which naturally generalizes the path concatenation operation by constructing all possible concatenations of path pairs from the given two sets. Finally, if a single-edge subpath is labeled by a terminal, then the corresponding edge should be added to the result and if the label is nonterminal, \textsc{getPaths} should be used to restore paths.

It is assumed that the sets are computed lazily, so as to ensure the termination in case of an infinite number of paths.
We also do not check paths for duplication explicitely, since they are assumed to be represented as sets.
\subsection{An example}
\label{example:section}
In this section, we introduce a detailed example to demonstrate the steps taken by the proposed algorithms.
Namely, consider the graph $\mathcal{G}$ presented in Figure~\ref{fig:example_input_graph} and the RSM $R$ presented in Figure~\ref{example:automata}.

In the first step, we represent both the graph and RSM as a set of Boolean matrices.
Notice that we should add a new empty matrix $M_2^{S}$ to $\mathcal{M}_2$,
where edges labeled by $S$ will be added at the time of the computation.

After the initialization, the algorithm handles the $\varepsilon$-case.
The input RSM does not have any $\varepsilon$-transitions and does not have any states that are both start and final, therefore, no edges are added at this stage.
After that, we should compute $\mathcal{M}_2$ and $C_3$ iteratively.
We denote the iteration number of the loop of matrices evaluation as the number in parentheses in the subscript.

\textbf{The first iteration.} First of all, we compute the Kronecker product of the
$\mathcal{M}_1$ and $\mathcal{M}_{2,(0)}$ matrices and collapse the result to the single Boolean matrix
$M_{3,(1)}$. For the sake of simplicity, we provide only
$M_{3,(1)}$, which is evaluated as follows.
{
    \renewcommand{\arraystretch}{0.5}
    \setlength\arraycolsep{0.1pt}
\begin{align*}
  \centering
& M_{3,(1)} = M_1^a \otimes M_{2,(0)}^a +  M_1^b \otimes M_{2,(0)}^b + M_1^S \otimes M_{2,(0)}^S =\\
& \kbordermatrix{
          & (0,0) & (0,1) & \vrule & (1,0) & (1,1) & \vrule &  (2,0) & (2,1) & \vrule &  (3,0) & (3,1) &\\
    (0,0) & . & .  & \vrule & . & 1  & \vrule & . & .  &  \vrule & . & .  \\
    (0,1) & . & .  & \vrule & 1 & .   & \vrule & . & .  &  \vrule & . & .  \\
    \hline
    (1,0) & . & .   & \vrule & . & .  & \vrule & . & .  & \vrule & . & . \\
    (1,1) & . & .   & \vrule & . & .  & \vrule & . & .  & \vrule & .  &1   \\
    \hline
    (2,0) & . & .   & \vrule & . & .  & \vrule & . & .  & \vrule & . & .  \\
    (2,1) & . & .   & \vrule & . & .  & \vrule & . & .  & \vrule & . &1  \\
    \hline
    (3,0) & . & .   & \vrule & . & .  & \vrule & . & .  & \vrule & . & .  \\
    (3,1) & . & .   & \vrule & . & .  & \vrule & . & .  & \vrule & . & .  \\
}
\end{align*}
}

As far as the input graph has no edges with label $S$, the correspondent block of the Kronecker product will be empty. The Kronecker product graph of the input graph $\mathcal{G}$ and RSM $R$ is shown in Figure~\ref{fig:example_1_product}. Then, the transitive closure evaluation result, stored in the matrix $C_{3,(1)}$, introduces one new path of length 2 (the thick edges in Figure~\ref{fig:example_1_product}).

\begin{figure}[h]
    \centering
   \begin{tikzpicture}[->,auto,node distance=0.5cm]
       \node[state] (q_0)                      {$(0, 0)$};
       \node[state] (q_1) [right=of q_0] {$(1, 0)$};
       \node[state] (q_2)  [right=of q_1] {$(2, 0)$};
       \node[state] (q_3) [right=of q_2] {$(3, 0)$};
       \node[state] (q_4)  [below=of q_0] {$(0, 1)$};
       \node[state] (q_5)  [right=of q_4] {$(1, 1)$};
      \node[state] (q_6)  [right=of q_5] {$(2, 1)$};
       \node[state] (q_7)  [right=of q_6] {$(3, 1)$};
       \path[->]
        (q_0) edge[very thick]  node {} (q_5)
        (q_4) edge  node {} (q_1)
        (q_5) edge [bend left, out=40, in=140, below, very thick]  node {} (q_7)
        (q_6) edge   node {} (q_7)
        ;
    \end{tikzpicture}
    \caption{The Kronecker product graph of RSM $R$ and the input graph $\mathcal{G}$ (edges which form new paths are thick)}
    \label{fig:example_1_product}
\end{figure}
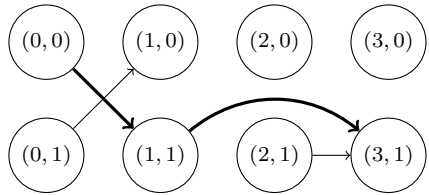

This path starts in the vertex $(0,0)$ and finishes in the vertex $(3,1)$.
We can see, that 0 and 3 are the start and final states of some component
state machine for label $S$ in $R$ respectively. Thus we can conclude that
there exists a path between vertices 0 and 1 in the graph, such that the
respective word is derivable from $S$ in the $R$ execution flow.

As a result, we can add the edge $(0,S,1)$ to the resulting graph, what is done by updating the matrix $M_2^S$.

\textbf{The second iteration.} The modified graph Boolean adjacency matrices contain
an edge with label $S$. Therefore, this label contributes to the non-empty
corresponding matrix block in the evaluated matrix $M_{3,{2}}$. The transitive closure
evaluation introduces three new paths $(0, 1) \rightarrow (2,1), (1, 0) \rightarrow (3,1)$ and $(0, 1) \rightarrow (3,1)$ (see Figure~\ref{fig:example_2_product}). Since only the path between vertices $(0,1)$ and
$(3,1)$ connects the start and final states in the automaton, the edge $(1,S,1)$ is added to the resulting graph.
{
    \renewcommand{\arraystretch}{0.5}
    \setlength\arraycolsep{0.1pt}
\begin{align*}
  \centering
& M_{3,(2)} = M_1^a \otimes M_{2,(2)}^a +  M_1^b \otimes M_{2,(2)}^b + M_1^S \otimes M_{2,(2)}^S = \\
& \kbordermatrix{
          & (0,0) & (0,1) & \vrule & (1,0) & (1,1) & \vrule &  (2,0) & (2,1) & \vrule &  (3,0) & (3,1) &\\
    (0,0) & . & .  & \vrule & . & 1  & \vrule & . & .  &  \vrule & . & .  \\
    (0,1) & . & .  & \vrule & 1 & .   & \vrule & . & .  &  \vrule & . & .  \\
    \hline
    (1,0) & . & .   & \vrule & . & .  & \vrule & . & \mc  & \vrule & . & . \\
    (1,1) & . & .   & \vrule & . & .  & \vrule & . & .  & \vrule & .  &1   \\
    \hline
    (2,0) & . & .   & \vrule & . & .  & \vrule & . & .  & \vrule & . & .  \\
    (2,1) & . & .   & \vrule & . & .  & \vrule & . & .  & \vrule & . &1  \\
    \hline
    (3,0) & . & .   & \vrule & . & .  & \vrule & . & .  & \vrule & . & .  \\
    (3,1) & . & .   & \vrule & . & .  & \vrule & . & .  & \vrule & . & .  \\
}
\end{align*}
}
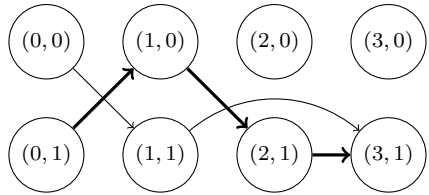
\begin{figure}[h]
    \centering
   \begin{tikzpicture}[->,auto,node distance=0.5cm]
       \node[state] (q_0)                      {$(0, 0)$};
       \node[state] (q_1) [right=of q_0] {$(1, 0)$};
       \node[state] (q_2)  [right=of q_1] {$(2, 0)$};
       \node[state] (q_3) [right=of q_2] {$(3, 0)$};
       \node[state] (q_4)  [below=of q_0] {$(0, 1)$};
       \node[state] (q_5)  [right=of q_4] {$(1, 1)$};
      \node[state] (q_6)  [right=of q_5] {$(2, 1)$};
       \node[state] (q_7)  [right=of q_6] {$(3, 1)$};
       \path[->]
        (q_1) edge[very thick] node {} (q_6)
        (q_0) edge node {} (q_5)
        (q_4) edge[very thick]  node {} (q_1)
        (q_5) edge [bend left, in=140, out=40, below]  node {} (q_7)
        (q_6) edge[very thick]   node {} (q_7)
        ;
    \end{tikzpicture}
    \caption{The Kronecker product graph of RSM $R$ and the updated graph $\mathcal{G}$ (edges which form new paths are thick)}
    \label{fig:example_2_product}
\end{figure}
The result graph is presented in Figure~\ref{fig:example_result}.
\begin{figure}[h]
    \centering
    \begin{tikzpicture}[shorten >=1pt,auto]
       \node[state] (q_0)                      {$0$};
       \node[state] (q_1) [right=of q_0]       {$1$};
        \path[->]
        (q_0) edge[bend left, above]   node {a} (q_1)
         (q_0) edge[in=190, out=-10, red, very thick]   node {S} (q_1)
        (q_1) edge [bend left, below] node {a} (q_0)
         (q_1) edge[loop above, red, thick] node {S} (q_1)
        (q_1) edge[loop right] node {b} (q_1);

    \end{tikzpicture}
    \caption{The result graph $\mathcal{G}$}
    \label{fig:example_result}
\end{figure}
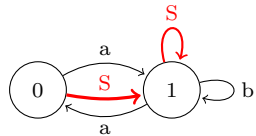

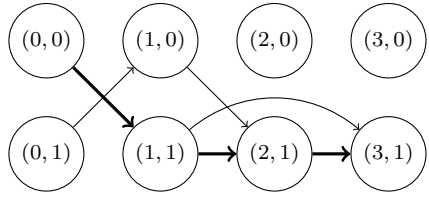
\begin{figure}[h]
	\centering
	\begin{tikzpicture}[->,auto,node distance=0.5cm]
		\node[state] (q_0)                      {$(0, 0)$};
		\node[state] (q_1) [right=of q_0] {$(1, 0)$};
		\node[state] (q_2)  [right=of q_1] {$(2, 0)$};
		\node[state] (q_3) [right=of q_2] {$(3, 0)$};
		\node[state] (q_4)  [below=of q_0] {$(0, 1)$};
		\node[state] (q_5)  [right=of q_4] {$(1, 1)$};
		\node[state] (q_6)  [right=of q_5] {$(2, 1)$};
		\node[state] (q_7)  [right=of q_6] {$(3, 1)$};
		\path[->]
		(q_1) edge node {} (q_6)
		(q_0) edge[very thick] node {} (q_5)
		(q_4) edge  node {} (q_1)
		(q_5) edge [bend left, in=140, out=40, below]  node {} (q_7)
		(q_5) edge[very thick]   node {} (q_6)
		(q_6) edge[very thick]   node {} (q_7)
		;
	\end{tikzpicture}
	\caption{The Kronecker product graph of RSM $R$ and the final graph $\mathcal{G}$ (edges which form new paths are thick)}
	\label{fig:example_3_product}
\end{figure}

No more edges will be added to the graph $\mathcal{G}$ at the last iteration.
However, the new edge $(1, 1) \rightarrow (2,1)$ will be added to the resulting Kronecker product, and the transitive closure
evaluation introduces one new path $(0, 0) \rightarrow (3,1)$ that connects the start and final states in the automaton (see Figure~\ref{fig:example_3_product}).
At this point, the index creation is finished.
One can use it to answer reachability queries, but it also can be used
to restore paths for some reachable vertices. The resulting Kronecker product matrix
$M_3$, or so-called \textit{index}, can be used for it.
For example, we can restore paths from vertex 1 to vertex 1 derived from $S$ in the resulting graph.

To get these paths we should call \verb|getPaths(1, 1, S)| function.
A partial trace of this call is presented in Figure~\ref{trc:example}.
First, we query paths for all possible start and final states of the
machine for the provided graph vertices.
Since the component state machine with label $S$ in the example RSM has the single final state, the function \verb|genPaths| is called with the arguments $(0,1)$ and $(3,1)$.
Note, that the values passed to the functions in the path extraction algorithm are the
pairs of the machine state and graph vertex, which uniquely identify a cell of
the index matrix $M_3$. As a result,
we get the set of all possible paths in the graph from $1$ to $1$ derived from $S$.

\begin{figure}
\begin{minipage}[t]{0.48\textwidth}
{
\scriptsize
\setlength{\DTbaselineskip}{8pt}
\DTsetlength{0.2em}{0.5em}{0.2em}{0.4pt}{1.6pt}
\dirtree{%
.1 getPaths($1,1,S$).
.2 genPaths($(0,1),(3,1)$).
.3 return $\{[((0,1),(1,0)), ((1,0),(2,1)), ((2,1),(3,1))]\}$.
.2 currentSubPaths = $\{[1 \xrightarrow{a} 0]\}$.
.2 currentPaths = $\{[(1 \xrightarrow{a} 0)]\}$.
.2 getPaths($0,1,S$).
.3 genPaths($(0,0),(3,1)$).
.4 return $\{[((0,0),(1,1)), ((1,1),(3,1))], [((0,0),(1,1)), ((1,1),(2,1)), ((2,1),(3,1))]\}$.
.3 path = $[((0,0),(1,1)), ((1,1),(3,1))]$.
.3 currentSubPaths = $\{[0 \xrightarrow{a} 1]\}$.
.3 currentPaths = $\{[0 \xrightarrow{a} 1]\}$.
.3 currentSubPaths = $\{[1 \xrightarrow{b} 1]\}$.
.3 $\cdots$.
.3 resultPaths = $\{[0 \xrightarrow{a} 1 \xrightarrow{b} 1]\}$.
.3 path = $[((0,0),(1,1)), ((1,1),(2,1)), ((2,1),(3,1))]$.
.3 currentSubPaths = $\{[0 \xrightarrow{a} 1]\}$.
.3 currentPaths = $\{[0 \xrightarrow{a} 1]\}$.
.3 getPaths(1, 1, S) // \begin{minipage}[t]{14cm} An alternative way to get paths from 1 to 1 (leads to infinite set of paths) \end{minipage}.
.3 $\cdots$.
.8 return $r_\infty^{1 \leadsto 1}$ // \begin{minipage}[t]{5cm} An infinite set of path from 1 to 1 \end{minipage}.
.3 currentPaths = $\{[0 \xrightarrow{a} 1] \} \cdot r_\infty^{1\leadsto 1}$.
.3 currentSubPaths = $\{[1 \xrightarrow{b} 1]\}$.
.3 $\cdots$.
.3 return $\{[0 \xrightarrow{a} 1 \xrightarrow{b} 1]\} \cup (\{[0 \xrightarrow{a} 1] \} \cdot r_\infty^{1\leadsto 1} \cdot \{[1 \xrightarrow{b} 1]\})$.
.2 currentPaths = $\{[1 \xrightarrow{a} 0 \xrightarrow{a} 1 \xrightarrow{b} 1]\} \cup (\{[1 \xrightarrow{a} 0 \xrightarrow{a} 1] \} \cdot r_\infty^{1\leadsto 1} \cdot \{[1 \xrightarrow{b} 1]\})$.
.2 currentSubPaths = $\{[1 \xrightarrow{b} 1]\}$.
.2 $\cdots$.
.2 return = $\{[1 \xrightarrow{a} 0 \xrightarrow{a} 1 \xrightarrow{b} 1 \xrightarrow{b} 1]\} \cup (\{[1 \xrightarrow{a} 0 \xrightarrow{a} 1] \} \cdot r_\infty^{1\leadsto 1} \cdot \{[1 \xrightarrow{b} 1 \xrightarrow{b} 1]\})$.
}
}
\caption{Example of call stack trace}
\label{trc:example}
\end{minipage}
\end{figure}
\section{Evaluation}

The goal of this evaluation is to investigate the applicability of the proposed algorithm to both regular and context-free path querying.
We measured the execution time of the index creation which solves the reachability problem for both kinds of queries.
The execution time for CFPQ was compared with Azimov's algorithm for CFPQ reachability.
We also investigated the practical applicability of the paths extraction algorithm to both regular and context-free path queries.

For evaluation, we used a PC with Ubuntu 18.04 installed.
It has Intel core i7-6700 CPU, 3.4GHz, and DDR4 64Gb RAM.
We only measure the execution time of the algorithms themselves, thus we assume an input graph is loaded into RAM in the form of its adjacency matrix in the sparse format.
Note, that the time needed to load an input graph into the RAM is excluded from the time measurements.

\subsection{RPQ Evaluation}

To investigate the applicability of the proposed algorithm for regular path querying we gathered a dataset that consists of both real-world and synthetically generated graphs.
We generated the queries from the most popular RPQ templates.

\subsubsection{Dataset}

We gathered several graphs that represent real-world data from different areas and are frequently used for the evaluation of the graph querying algorithms.
Namely, the dataset consists of three parts.
The first part is the set of LUBM graphs\footnote{Lehigh University Benchmark (LUBM) web page: \url{http://swat.cse.lehigh.edu/projects/lubm/}. Access date: 07.07.2020.}~\citep{10.1016/j.websem.2005.06.005} which have different numbers of vertices.
The second one is the set of graphs from Uniprot database\footnote{Universal Protein Resource (UniProt) web page: \url{https://www.uniprot.org/}. All files used can be downloaded via the link: \url{ftp://ftp.uniprot.org/pub/databases/uniprot/current_release/rdf/}. Access date: 07.07.2020.}: \textit{proteomes}, \textit{taxonomy} and \textit{uniprotkb}.
The~last part consists of the RDF files \textit{mappingbased\_properties} from DBpedia\footnote{DBpedia project web site: \url{https://wiki.dbpedia.org/}. Access date: 07.07.2020.} and \textit{geospecies}\footnote{The Geospecies RDF: \url{https://old.datahub.io/dataset/geospecies}. Access date: 07.07.2020.}.
A brief description of the graphs in the dataset is presented in Table~\ref{tbl:graphs_for_rpq}.

\begin{table}
    \centering
\caption{Graphs for RPQ evaluation}
\label{tbl:graphs_for_rpq}
{

\rowcolors{2}{black!2}{black!10}
\begin{tabular}{|l|c|c|}
\hline
Graph & \#V & \#E  \\
\hline
\hline
LUBM1k  & 120 926 & 484 646 \\
LUBM3.5k  & 358 434 & 144 9711 \\
LUBM5.9k  & 596 760 & 2 416 513 \\
LUBM1M   & 1 188 340 & 4 820 728 \\
LUBM1.7M & 1 780 956 & 7 228 358 \\
LUBM2.3M & 2 308 385 & 9 369 511 \\
\hline
Uniprotkb & 6 442 630 & 24 465 430 \\
Proteomes & 4 834 262 & 12 366 973 \\
Taxonomy & 5 728 398 & 14 922 125 \\
\hline
Geospecies & 450 609 & 2 201 532 \\
Mappingbased\_properties & 8 332 233 & 25 346 359 \\
\hline
\end{tabular}
}
\end{table}

Queries for evaluation were generated from the templates for the most popular RPQs, specifically, the queries presented in Table 2 in~\cite{Pacaci2020RegularPQ} and in Table 5 in~\cite{Wang2019}.
These query templates are presented in Table~\ref{tbl:queries_templates}.
We generate 10 queries for each template and each graph.
The most frequent relations from the given graph were used as symbols in the query template\footnote{Used generator is available as part of CFPQ\_data project: \url{https://github.com/JetBrains-Research/CFPQ_Data/blob/master/tools/gen_RPQ/gen.py}. Access data: 07.07.2020.}.
We used the same set of queries for all LUBM graphs to investigate the scalability of the proposed algorithm.

\begin{table}
    \centering
\caption{Queries templates for RPQ evaluation}
\label{tbl:queries_templates}
{\small
\renewcommand{\arraystretch}{1.2}
\rowcolors{2}{black!2}{black!10}
\begin{tabular}{|c|c||c|c|}
\hline

Name & Query & Name & Query \\
\hline
\hline
$Q_1$   & $a^*$                               & $Q_9^5$    & $(a \mid b \mid c \mid d \mid e)^+$                     \\
$Q_2$   & $a\cdot b^*$                        & $Q_{10}^2$ & $(a \mid b) \cdot c^*$                                  \\
$Q_3$   & $a \cdot b^* \cdot c^*$             & $Q_{10}^3$ & $(a \mid b \mid c)  \cdot d^*$                          \\
$Q_4^2$ & $(a \mid b)^*$                      & $Q_{10}^4$ & $(a \mid b \mid c \mid d)  \cdot e^*$                   \\
$Q_4^3$ & $(a \mid b \mid c)^*$               & $Q_{10}^5$ & $(a \mid b \mid c \mid d \mid e)  \cdot f^*$            \\
$Q_4^4$ & $(a \mid b \mid c \mid d)^*$        & $Q_{10}^2$ & $a \cdot b$                                             \\
$Q_4^5$ & $(a \mid b \mid c \mid d \mid e)^*$ & $Q_{11}^3$ & $a \cdot b \cdot c$                                     \\
$Q_5$   & $a \cdot b^* \cdot c$               & $Q_{11}^4$ & $a \cdot b \cdot c \cdot d$                             \\
$Q_6$   & $a^* \cdot b^*$                     & $Q_{11}^5$ & $a \cdot b \cdot c \cdot d \cdot f$                     \\
$Q_7$   & $a \cdot b \cdot c^*$               & $Q_{12}$   & $(a \cdot b)^+ \mid  (c \cdot d)^+$                     \\
$Q_8$   & $a? \cdot b^*$                      & $Q_{13}$   & $(a \cdot(b \cdot c)^*)^+ \mid  (d \cdot f)^+$          \\
$Q_9^2$ & $(a \mid b)^+$                      & $Q_{14}$   & $(a \cdot b \cdot (c \cdot d)^*)^+  \cdot (e \mid f)^*$ \\
$Q_9^3$ & $(a \mid b \mid c)^+$               & $Q_{15}$   & $(a \mid b)^+ \cdot (c \mid d)^+$                       \\
$Q_9^4$ & $(a \mid b \mid c \mid d)^+$        & $Q_{16}$   & $a \cdot b \cdot (c \mid d \mid e)$                     \\
\hline
\end{tabular}
}
\end{table}

\subsubsection{Results}

We averaged the execution time of index creation over 5 runs for each query.
Index creation time for LUBM graphs set is presented in Figure~\ref{fig:lubm_all_qs}.
We can see that evaluation time depends on the query: there are queries evaluated in less than 1 second even for the largest graphs ($Q_2$, $Q_5$, $Q_{11}^2$, $Q_{11}^3$), while the worst time is 6.26 seconds ($Q_{14}$).
The execution time of our algorithm is comparable with the recent results for the same graphs and queries implemented on a distributed system over 10 nodes~\citep{Wang2019}, while we use only one node.
We conclude that our algorithm demonstrates reasonable performance to be applied to the real-world data analysis.

\begin{figure}
    \centering
   \includegraphics[width=0.6\textwidth]{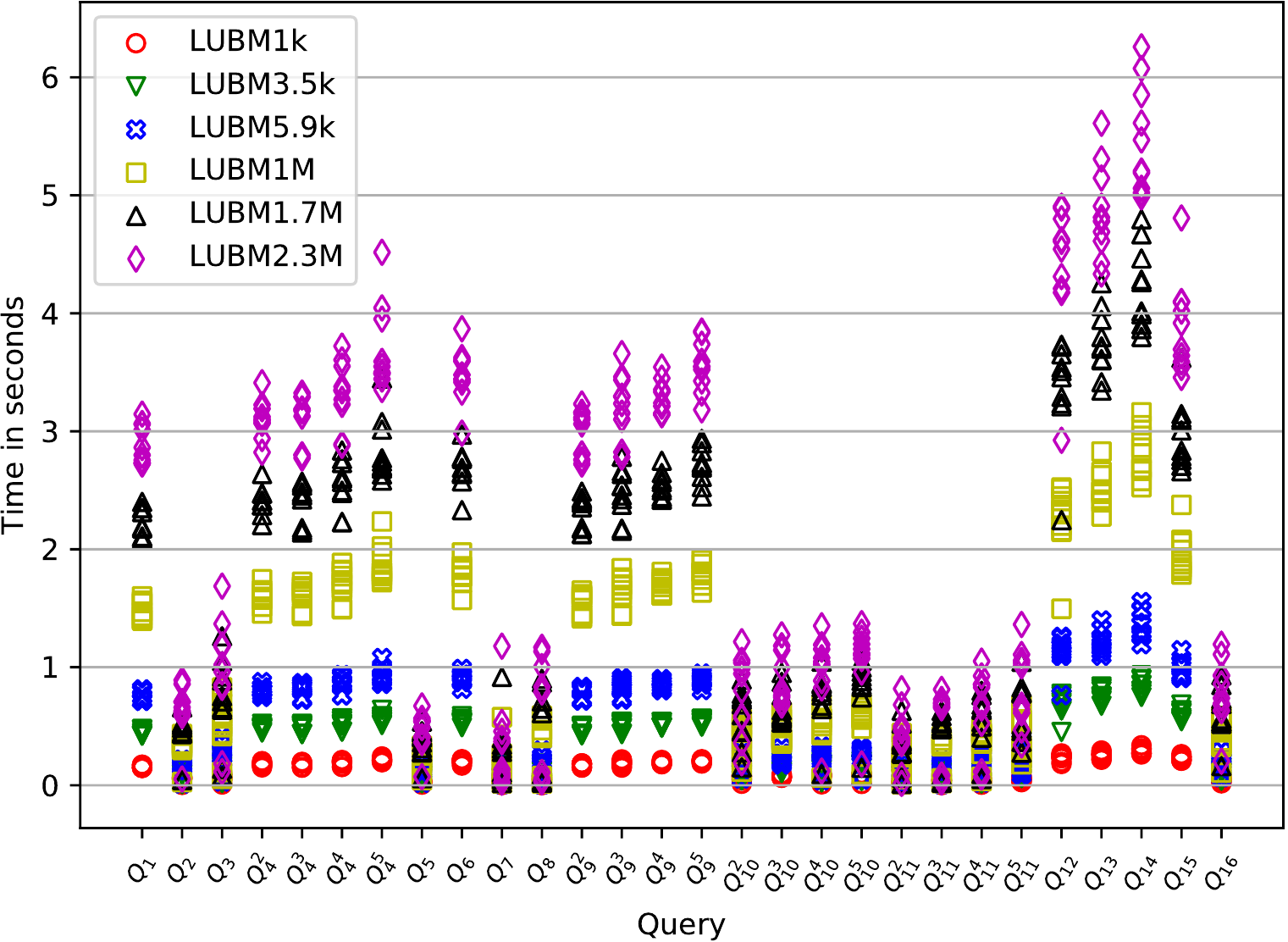}
   \caption{Index creation time for LUBM graphs}
   \label{fig:lubm_all_qs}
\end{figure}

Index creation time for each query on the real-world graphs is presented in Figure~\ref{fig:other_all_qs}.
We can see that querying small graphs requires more time than querying bigger graphs in some cases.
For example, consider $Q_{10}^4$: querying the \textit{geospecies} graph (450k vertices) in some cases requires more time than querying of \textit{mappingbased\_properties} (8.3M vertices) and \textit{taxonomy} (5.7M vertices).
We conclude that the evaluation time depends on the inner structure of a graph.
On the other hand, \textit{taxonomy} querying in many cases requires significantly more time than for other graphs, while \textit{taxonomy} is not the biggest graph.
Finally, in most cases, query execution lasts less than 10 seconds, even for bigger graphs, and no query requires more than 52.17 seconds.

\begin{figure}
    \centering
   \includegraphics[width=0.6\textwidth]{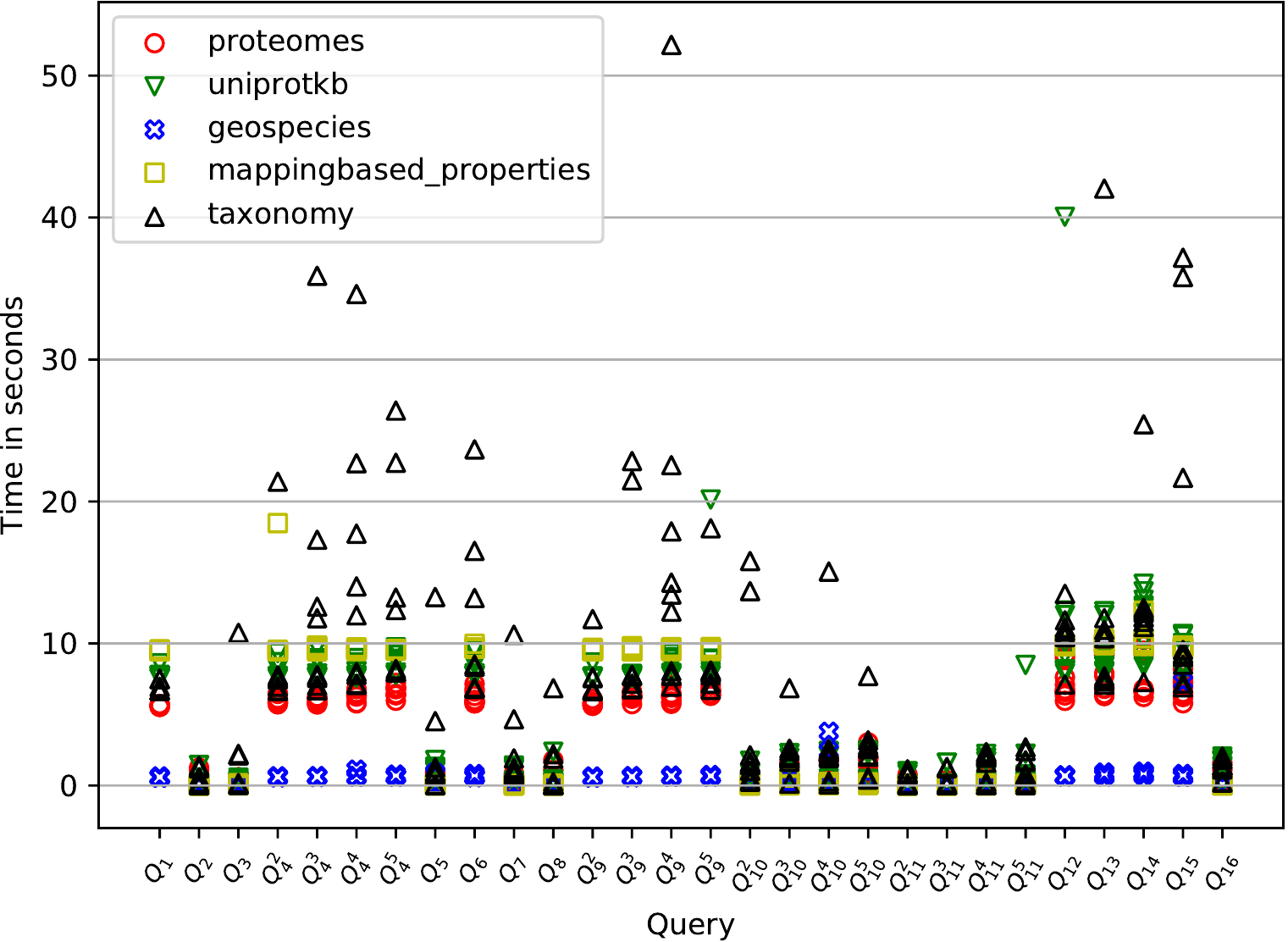}
   \caption{Index creation time for real-world RDFs}
   \label{fig:other_all_qs}
\end{figure}

\subsection{CFPQ Evaluation}

We evaluate the applicability of the proposed algorithm to CFPQ processing over real-world graphs on a number of classic cases and compare them with Azimov's algorithm.
Currently, only a single path version of Azimov's algorithm exists, and we use its implementation using PyGraphBLAS. Note that it is not trivial to compare our results with the state-of-the-art results provided by~\cite{10.1145/3398682.3399163} (Azimov's algorithm) because our algorithm computes significantly more information. While the state-of-the-art solution computes only reachability facts or a single-path semantics, our algorithm computes data necessary to restore all possible paths.

\subsubsection{Dataset}

We use CFPQ\_Data\footnote{CFPQ\_Data is a dataset for CFPQ evaluation which contains both synthetic and real-world data and queries \url{https://github.com/JetBrains-Research/CFPQ\_Data}. Access date: 07.07.2020.} dataset for evaluation.
Namely, we use relatively big RDF files and respective same-generation queries $G_1$~(Eq.~\ref{eqn:g_1}) and $G_2$~(Eq.~\ref{eqn:g_2}) which are used in other works for CFPQ evaluation.
We also use the $Geo$~(Eq.~\ref{eqn:geo}) query provided by~\cite{Kuijpers:2019:ESC:3335783.3335791} for \textit{geospecies} RDF.
Note that we use $\overline{x}$ notation in queries to denote the inverse of $x$ relation and the respective edge.
\begin{align}
\begin{split}
\label{eqn:g_1}
S \to & \overline{\textit{subClassOf}} \ \ S \ \textit{subClasOf} \mid \overline{\textit{type}} \ \ S \ \textit{type}\\   & \mid \overline{\textit{subClassOf}} \ \ \textit{subClasOf} \mid \overline{\textit{type}} \ \textit{type}
\end{split}
\end{align}
\begin{align}
\begin{split}
\label{eqn:g_2}
S \to \overline{\textit{subClassOf}} \ \ S \ \textit{subClasOf} \mid \textit{subClassOf}
\end{split}
\end{align}
\begin{align}
\begin{split}
\label{eqn:geo}
S \to & \textit{broaderTransitive} \ \  S \ \overline{\textit{broaderTransitive}} \\
      & \mid \textit{broaderTransitive} \ \  \overline{\textit{broaderTransitive}}
\end{split}
\end{align}
\begin{align}
\begin{split}
\label{eqn:ma}
S & \to \overline{d} \ V \ d \\
V & \to ((S?) \overline{a})^* (S?) (a (S?))^*
\end{split}
\end{align}

Additionally, we evaluate our algorithm on memory aliases analysis problem: a well-known problem which can be reduced to CFPQ~\citep{Zheng:2008:DAA:1328897.1328464}.
To do it, we use some graphs built for different parts of Linux OS kernel (\textit{arch}, \textit{crypto}, \textit{drivers}, \textit{fs}) and the query \textit{MA}~(Eq.~\ref{eqn:ma})~\citep{10.1145/3093336.3037744}.
The detailed data about all the graphs used is presented in Table~\ref{tbl:graphs_for_cfpq}.

{\setlength{\tabcolsep}{0.3em}
\begin{table}
    \centering
{
\caption{Graphs for CFPQ evaluation: \textit{bt} is broaderTransitive, \textit{sco} is subCalssOf}
\label{tbl:graphs_for_cfpq}
\scriptsize
\rowcolors{2}{black!2}{black!10}
\begin{tabular}{|l|c|c|c|c|c|c|c|}
\hline
Graph          & \#V       & \#E        & \#sco & \#type &\#bt & \#a  & \#d \\
\hline
\hline
eclass\_514en  & 239 111    & 523 727    & 90 512    & 72 517    &        ---        & ---  & --- \\
enzyme         & 48 815     & 109 695    & 8 163     & 14 989    &        ---        & ---  & --- \\
geospecies     & 450 609    & 2 201 532  & 0         & 89 062    &        20 867     & ---  & --- \\
go             & 272 770    & 534 311    & 90 512    & 58 483    &        ---        & ---  & --- \\
go-hierarchy   & 45 007     & 980 218    & 490 109   & 0         &        ---        & ---  & --- \\
taxonomy       & 5 728 398  & 14 922 125 & 2 112 637 & 2 508 635 &        ---        & ---  & --- \\
\hline
arch           & 3 448 422  & 5 940 484  &      ---     &  ---   &        ---        & 671 295 & 2 298 947 \\
crypto         & 3 464 970  & 5 976 774  &      ---     &  ---   &        ---        & 678 408 & 2 309 979 \\
drivers        & 4 273 803  & 7 415 538  &      ---     &  ---   &        ---        & 858 568 & 2 849 201 \\
fs             & 4 177 416  & 7 218 746  &      ---     &  ---   &        ---        & 824 430 & 2 784 943 \\
\hline
\end{tabular}
}
\end{table}
}
\subsubsection{Results}

We averaged the index creation time over 5 runs for both single-path Azimov's algorithm (\textbf{Mtx}) and the proposed algorithm (\textbf{Tns}) (see Table~\ref{tbl:CFPQ_index}).

{\setlength{\tabcolsep}{0.2em}
  \begin{table}
    \centering
    \caption{CFPQ evaluation results, time is measured in seconds}
    \label{tbl:CFPQ_index}
    \rowcolors{4}{black!2}{black!10}
    \small
    \begin{tabular}{| l | c | c | c | c | c | c | c | c |}
      \hline

      \multirow{2}{*}{Name}  & \multicolumn{2}{c|}{$G_1$} & \multicolumn{2}{c|}{$G_2$} & \multicolumn{2}{c|}{\textit{Geo}} & \multicolumn{2}{c|}{\textit{MA}}\\
      \cline{2-9}
                      & Tns    & Mtx    & Tns  & Mtx  & Tns   & Mtx   & Tns     & Mtx \\
      \hline
      \hline
      eclass\_514en   & 0.24   & 0.27   & 0.25 & 0.26 & ---   & ---   & ---     & ---\\
      enzyme          & 0.03   & 0.04   & 0.02 & 0.01 & ---   & ---   & ---     & ---\\
      geospecies      & 0.08   & 0.06   & $<0.01$ & 0.01 & 26.12 & 16.58 & ---     & ---\\
      go-hierarchy    & 0.16   & 1.43   & 0.23 & 0.86 & ---   & ---   & ---     & ---\\
      go              & 1.56   & 1.74   & 1.21 & 1.14 & ---   & ---   & ---     & ---\\
      pathways        & 0.01   & 0.01   & 0.01 & 0.01 & ---   & ---   & ---     & ---\\
      taxonomy        & 4.81   & 2.71   & 3.75 & 1.56 & ---   & ---   & ---     & ---\\
      \hline
      arch            & ---    & ---    & ---  & ---  & ---   & ---   & 262.45  & 195.51  \\
      crypto          & ---    & ---    & ---  & ---  & ---   & ---   & 267.52  & 195.54  \\
      drivers         & ---    & ---    & ---  & ---  & ---   & ---   & 1309.57 & 1050.78 \\
      fs              & ---    & ---    & ---  & ---  & ---   & ---   & 470.49  & 370.73  \\
      \hline
    \end{tabular}
  \end{table}
}

We can see that while in some cases our solution is comparable or just slightly better than Azimov's algorithm (\textit{enzyme, eclass\_514en, go}), there are cases when our solution is significantly faster (\textit{go-hierarchy}, up to 9 times faster), and when Azimov's algorithm about 1.3 times faster (all memory aliases and \textit{geospecies} with \textit{Geo} query).
Thus we can conclude that our solution is performant enough because Azimov's algorithm is the fastest known practical CFPQ algorithm~\citep{10.1145/3398682.3399163} and the current version, which we use for evaluation, computes index only for single-path semantics, but our algorithm computes much more information (index for all paths) in a comparable time.

Best to our knowledge, the proposed algorithm is the first algorithm that provides information about all paths of interest (Azimov's algorithm computes information about only one path).
The direct comparison with other solutions is impossible, and we just estimate the running time of our algorithm for a small number of cases.
Namely, we extract all paths with length not greater than 20 edges between all pairs of vertices from indices created for graphs \textit{go} and \textit{eclass\_514en} and query $G_1$.
Paths extraction for one pair of vertices requires 2.64 seconds averaged over all pairs for \textit{go} graph. The~maximal time is 4699 seconds and 217 737 paths were extracted during this time. The average number of paths between two vertices is 184.
For \textit{eclass\_514en} paths extraction for one pair of vertices requires 1.27 seconds averaged over all pairs. The~maximal time is 8.04 seconds and only one path is extracted during this time. The average number of paths between two vertices is 3.
We can see that paths can be extracted in a reasonable time, but a detailed analysis of paths extraction algorithm performance depends on graphs structure.


\subsection{Conclusion}

We conclude that the proposed algorithm is applicable to real-world data processing: the algorithm allows one to solve both the reachability problem and to extract paths of interest in a reasonable time.
While index creation time (reachability query evaluation) is comparable with other existing solutions, the paths extraction procedure should be improved in the future. However, the state-of-the-art solution computes only reachability facts or a single-path semantics, whereas our algorithm computes data necessary to restore all possible paths (all-paths semantics).
Finally, a detailed comparison of the proposed solution with other algorithms for CFPQ and RPQ is required.

To summarize the overall evaluation, the proposed algorithm is applicable to both RPQ and CFPQ over real-world graphs.
Thus, the proposed solution is a promising unified algorithm for both RPQ and CFPQ evaluation.

\section{Related Work}

Language constrained path querying is widely used in graph databases, static code analysis, and other areas.
Both, RPQ and CFPQ (known as CFL-reachability problem in static code analysis) are actively studied in the recent years.

There is a huge number of theoretical research on RPQ and its specific cases.
RPQ with single-path semantics was investigated from the theoretical point of view by~\cite{barrett2000formal}.
In order to research practical limits and restrictions of RPQ, a number of high-performance RPQ algorithms were provided.
For example, the derivative-based solution provided by~\cite{10.1145/2949689.2949711}, which is implemented on top of the Pregel-based system, or the solution by~\cite{10.1007/978-3-642-31235-9_12}.
But only a limited number of practical solutions provide the ability to restore paths of interest.
A recent work by~\cite{Wang2019} provides a Pregel-based provenance-aware RPQ algorithm which utilizes a Glushkov's construction~\citep{Glushkov1961}.
There is a lack of research of the applicability of linear algebra-based RPQ algorithms with paths-providing semantics.

On the other hand, many CFPQ algorithms with various properties were proposed recently.
They employ the ideas of different parsing algorithms, such as CYK in works by~\cite{hellingsRelational} and~\cite{8249039}, (G)LR and (G)LL in works by~\cite{Grigorev:2017:CPQ:3166094.3166104},~\cite{Medeiros:2018:EEC:3167132.3167265},~\cite{10.1007/978-3-319-91662-0_17},~\cite{10.1007/978-3-319-41579-6_22}.
Unfortunately, none of them has better than cubic time complexity in terms of the input graph size.
The algorithm by~\cite{Azimov:2018:CPQ:3210259.3210264} is, best to our knowledge, the first algorithm for CFPQ which is based on linear algebra.
It was shown by~\cite{10.1145/3398682.3399163} that this algorithm can be applied to real-world graph analysis problems, while~\cite{Kuijpers:2019:ESC:3335783.3335791} show that other state-of-the-art CFPQ algorithms are not performant enough to handle real-world graphs.

It is important in both RPQ and CFPQ to be able to restore paths of interest.
Some of the mentioned algorithms can solve only the reachability problem, while it may be important to provide at least one path which satisfies the query.
While~\cite{10.1145/3398682.3399163} provide the first CFPQ algorithm with single path semantics based on linear algebra,~\cite{HellSinglePath} provides the first theoretical investigation of this problem.
He also provides an overview of the related works and shows that the problem is related to the string generation problem and respective results from the formal language theory.
He concludes that both theoretical and empirical investigation of CFPQ with single-path and all-path semantics are at the early stage.
We agree with this point of view, and we only demonstrate the applicability of our solution to paths extraction and do not investigate its properties in details.

While CFPQ on $n$-node graph has a relatively straightforward $O(n^3)$ time algorithm, it is a long-standing open problem whether there exists a truly  subcubic $O(n^{3-\varepsilon})$ algorithm for this problem.
The question on the existence of a subcubic CFPQ algorithm was stated by~\cite{Yannakakis}.
A bit later~\cite{10.5555/271338.271343} proposed the CFL reachability as a framework for interprocedural static code analysis.
\cite{10.1145/258994.259006} gave a dynamic programming formulation of the problem which runs in $O(n^3)$ time.
The problem of the cubic bottleneck of context-free language reachability is also discussed by~\cite{10.5555/788019.788876} and~\cite{10.1145/258994.259006}.
The slightly subcubic algorithm with $O(n^3/\log{n})$ time complexity was provided by~\cite{10.1145/1328438.1328460}.
This result is inspired by recursive state machine reachability.
The first truly subcubic algorithm with $O(n^\omega polylog(n))$ time complexity ($\omega$ is the best exponent for matrix multiplication) for an arbitrary graph and 1-Dyck language was provided by~\cite{8249039}, and~\cite{pavlogiannis2020finegrained}.
Other partial cases were investigated by~\cite{10.1145/3158118},~\cite{zhang2020conditional}.

Employing linear algebra is a promising way to high-performance graph analysis.
There are many works which formulate specific graph algorithms in terms of linear algebra, for example, such algorithms as for computing transitive closure and all-pairs shortest paths.
Recently this direction was summarized in GrpahBLAS API~\citep{7761646} which provides building blocks to develop a graph analysis algorithm in terms of linear algebra.
There is a number of implementations of this API, such as SuiteSparse:GraphBLAS~\citep{10.1145/3322125} or CombBLAS~\citep{10.1177/1094342011403516}.
Approaches to evaluate different classes of queries in different systems based on linear algebra are being actively researched.
This approach demonstrates significant performance improvement when applied for SPARQL queries evaluation~\citep{10.1145/3302424.3303962,DBLP:journals/corr/MetzlerM15a} and for Datalog queries evaluation~\citep{sato_2017}.
Finally, RedisGraph~\citep{8778293}, a linear-algebra powered graph database, was created and it was shown that in some scenarios it outperforms many other graph databases.
\section{Conclusion and Future Work}

In this work, we present an improved version of the tensor-based algorithm for CFPQ: we reduce the algorithm to operations over Boolean matrices, and we provide the ability to extract all paths which satisfy the query.
Moreover, the provided algorithm can handle grammars in EBNF, thus it does not require CNF transformation of the grammar and avoids grammar size blow-up.
As a result, the algorithm demonstrates practical performance not only on CFPQ queries but also on RPQ ones, which is shown by our evaluation.
Thus, we provide a universal linear algebra based algorithm for RPQ and CFPQ evaluation with all-paths semantics.
Moreover, our algorithm opens a way to tackle the long-standing problem of determining whether a subcubic CFPQ exists.
This is done by reducing the algorithm to incremental transitive closure: incremental transitive closure with $O(n^{3-\varepsilon})$ total update time for $n^2$ updates, such that each update returns all of the new reachable pairs, implies $O(n^{3-\varepsilon})$ CFPQ algorithm.
We prove $O({|P|}^3n^3/\log (|P|n))$ time complexity by providing $O(n^3/\log{n})$ incremental transitive closure algorithm.

Recent hardness results for dynamic graph problems demonstrate that any further improvement for incremental transitive closure (and, hence, CFPQ) will imply a major breakthrough for other long-standing dynamic graph problems. An algorithm for incremental dynamic transitive closure with total update time $O(mn^{1-\varepsilon})$ ($n$ denotes the number of graph vertices, $m$ is the number of graph edges) even with polynomial $poly(n)$ time preprocessing of the input graph and $m^{\delta - \varepsilon}$ query time per query for any $\delta \in (0, 1/2]$ will refute the Online Boolean Matrix-Vector Multiplication (OMv) Conjecture, which is used to prove conditional lower bounds for many dynamic problems ~\citep{8948597, 10.1145/2746539.2746609}.

\cite{10.1145/3158118} obtained a conditional cubic lower bound for the CFL-reachability problem via the combinatorial BMM Conjecture. The combinatorial BMM Conjecture states that there is no truly subcubic $O(n^{3-\varepsilon})$ combinatorial algorithm for multiplying two $n \times n$ Boolean matrices. Such a lower bound holds only for \textit{combinatorial} algorithms, leaving the possibility of an improvement via
an algebraic algorithm. Can the above mentioned OMv Conjecture be used to establish conditional \textit{algebraic} lower bounds for the CFL-reachability (CFPQ evaluation) problem? Another powerful assumption relative to which many conditional lower bounds are proved is the strong exponential time hypothesis (SETH) introduced by \cite{IMPAGLIAZZO2001367}. Recently~\cite{chistikov2021subcubic} showed that there cannot be a fine-grained reduction from SAT to CFL-reachability for a conditional lower bound stronger than $n^{\omega}$, unless the nondeterministic strong exponential time hypothesis (NSETH) fails. Strong assumptions like 3SUM Conjecture by \cite{Gajentaan1995OnAC}, multiphase problem by \cite{Patrascu} are variants of OMv conjecture or can be strengthened through it \citep{10.1145/2746539.2746609}. So the OMv Conjecture seems to be a good candidate for a proving conditional lower bound for the CFL-reachability (CFPQ evaluation) problem.

Also, an interesting task for the future is to improve the logarithmic factor in the obtained bound.

We also plan to improve bounds in partial cases for which dynamic transitive closure can be computed faster than in the general case.
Examples of such cases are querying over planar graphs~\citep{10.1007/3-540-57273-2_72}, undirected graph, and others.
In the case of planarity, it is interesting to investigate properties of the input graph and grammar which allow us to preserve planarity during query evaluation.

Note that both our algorithm and the state-of-the-art solutions have subcubic time complexity in terms of the grammar size \citep{Azimov:2018:CPQ:3210259.3210264, hellingsRelational, hellingsPathQuerying, 10.1145/258994.259006, 10.1145/3398682.3399163}.
However, our approach does not require the input grammar to be translated to the Chomsky Normal Form and thus avoids the quadratic blow-up in its size which leads to a better performance in practice.

An important task for future research is a detailed investigation of the paths extraction algorithm.~\cite{HellSinglePath} provides a theoretical investigation of the single-path extraction and shows that the problem is related to the formal language theory.
Extraction of all paths is more complicated and should be investigated carefully in order to provide an optimal algorithm.


From a practical perspective, it is necessary to analyze the usability of advanced algorithms for dynamic transitive closure.
In the current work, we evaluate the naive implementation in which transitive closure is recalculated from scratch on each iteration.
It is shown by~\cite{cs6345} that some advanced algorithms for dynamic transitive closure can be efficiently implemented.
Can one of these algorithms be efficiently parallelized and utilized in the proposed algorithm?

Also, it is necessary to evaluate GPGPU-based implementation.
Evaluation of Azimov's algorithm shows that it is possible to improve performance by using GPGPU because operations of linear algebra can be efficiently implemented on GPGPU~\citep{Mishin:2019:ECP:3327964.3328503,10.1145/3398682.3399163}.
Moreover, for practical reasons, it is interesting to provide a multi-GPU version of the algorithm and to utilize unified memory, which is suitable for linear algebra based processing of out-of-GPGPU-memory data and traversing on large graphs~\citep{8946118,10.14778/3384345.3384358}.

In order to simplify the distributed processing of huge graphs, it may be necessary to investigate different formats for sparse matrices, such as HiCOO format~\citep{10.5555/3291656.3291682}.
Another interesting question in this direction is about the utilization of virtualization techniques: should we implement a distributed version of the algorithm manually or it can be better to use CPU and RAM virtualization to get a virtual machine with a huge amount of RAM and a big number of computational cores.
The experience of the Trinity project team shows that it can make sense~\citep{10.1145/2463676.2467799}.

Finally, it is necessary to provide a multiple-source version of the algorithm and integrate it with a graph database.
RedisGraph\footnote{RedisGraph is a graph database that is based on the Property Graph
Model. Project web page: \url{https://oss.redislabs.com/redisgraph/}. Access date:
07.07.2020.}~\citep{8778293} is a suitable candidate for this purpose.
This database uses SuiteSparse---an implementation of GraphBLAS standard---as a base for graph processing.
This fact allowed to~\cite{10.1145/3398682.3399163} to integrate Azimov's algorithm to RedisGraph with minimal effort.

%
%


\bibliographystyle{spbasic}
\bibliography{tensor_product_CFPQ}
\end{document}